\documentclass[a4paper]{jpconf}
\usepackage{graphicx}
\usepackage{hyperref}
\usepackage{amssymb}
\usepackage{epsfig}
\usepackage{epstopdf}
\usepackage{amstext}
\usepackage{caption}
\usepackage{comment}
\usepackage{amsmath,amssymb}
\usepackage{notoccite}
\usepackage{mathtools}
\usepackage{amssymb}
\usepackage{titlesec}
\usepackage[utf8]{inputenc}
\usepackage{amsthm}
\usepackage{xcolor}
\usepackage[colorinlistoftodos]{todonotes}

\usepackage{subcaption}

\usepackage{fancyhdr}

\newtheorem{theorem}{Theorem}[section]

\newtheorem{lemma}[theorem]{Lemma}

\newcommand {\ie} {\emph{i.e.,~}}
\newcommand {\cf} {\emph{cf.}}

\titleformat{\section}
  {\normalfont\fontsize{12}{17}\sffamily\bfseries}
  {\thesection}
  {1em}
  {}

\titleformat{\subsection}
  {\normalfont\fontsize{12}{17}\sffamily\bfseries\slshape}
  {\thesubsection}
  {1em}
  {}

\titleformat{\subsubsection}
  {\normalfont\fontsize{12}{17}\sffamily\bfseries\slshape}
  {\thesubsubsection}
  {1em}
  {}
\newtheorem{remark}{Remark}

\newcommand{\keywords}[1]{%
  \vspace{-1.5em}
  \hspace{6.2em}
  \noindent\textbf{Keywords: }#1
}

\begin{document}
\title{Strong overlap of deterministic and stochastic dynamics in a
super-diffusive regime}

\author{Muhammad Tayyab \textsuperscript{1} \& Jahanzeb Tariq \textsuperscript{2}}

\address{$^1$ School of Computing Sciences, Pak-Austria Fachhochschule Institute of Applied Sciences and
Technology, Mang, 22621 Haripur, Pakistan} 
\address{$^2$ School of Physics, Engineering and Computer Science, University of Hertfordshire, Hatfield, Hertfordshire AL10 9AB, United Kingdom}

\ead{muhammad.tayyab@fecid.paf-iast.edu.pk}
%\ead{jtabs@herts.ac.uk}

\begin{abstract}
%We consider deterministic dynamics, known as the slicer map (SM), which exhibits normal and anomalous diffusion by varying a single parameter. The statistics of the position moments and the low-order position autocorrelation function (PACF) of the SM closely overlap with those of a stochastic process called the L\'evy-Lorentz gas (LLg), particularly in the normal and strongly superdiffusive anomalous regimes. However, these statistics alone are insufficient to fully characterize the indistinguishability of the dynamics. To demonstrate how these dynamics strongly overlap, we must understand the scaling of the higher-order PACF. This is because it is necessary to perform realistic and reliable models of complex systems in areas such as statistical mechanics. This paper presents the scaling equivalence of the higher-order PACF for these two dynamics. For this purpose, we analytically derive the generalized PACF of the SM and some of its scaling forms by examining various relationships between times. More specifically, we also derive several scalings of the $3$-point PACF by intriguing several relations between the three times. We then compare these scalings with the power-law tails of the numerically estimated $3$-point PACF of the LLg. This comparison provides a detailed description of the correlation scalings of the SM, demonstrating that the SM shares key features with the LLg. Thus, the SM can serve as a model to gain deep insights into this more complex stochastic system.
%Refined
We consider deterministic dynamics, known as the slicer map (SM), which exhibits normal and anomalous diffusion by varying a single parameter. The statistics of the position moments and the low-order position autocorrelation function (PACF) of the SM closely overlap with those of a stochastic process called the L\'evy-Lorentz gas (LLg), particularly in the normal and strongly superdiffusive anomalous regimes. However, matching low-order statistics alone cannot fully characterize the microscopic dynamics or distinguish underlying process classes.
%However, these similarities alone are insufficient to fully characterize the dynamics. 
To demonstrate how these dynamics strongly overlap, we focus on the scaling of higher-order PACF, which provides a more detailed characterization. 
%This is crucial for building realistic models of complex systems in areas like statistical mechanics.
In this paper, we analytically derive the generalized PACF of the SM and explore its scaling forms under different temporal relationships. Specifically, we derive several scalings of the $3$-point PACF by analyzing intriguing relations between three times. We compare these scalings with the power-law tails of the numerically estimated $3$-point PACF of the LLg. This comparison provides a detailed description of the correlation scalings of the SM, demonstrating that the SM shares key features with the LLg. Our findings establish the SM as a deterministic analog of the LLg, enabling efficient prediction of multi-time correlations in superdiffusive systems.
%Our findings establish the SM as a valuable deterministic model to gain deeper insights into the more complex stochastic system of the LLg.
\end{abstract}
%%%%%%%%%%%%%%%%%%%%%%%%%%%%%%%%%%%%%%%%%%%%%%%%%%%%%%%%%%%%%%%%%%%%%%%%%%%%%%%%%%%%%%%%%%%%%%%%%%%%
\section{Introduction}
To study anomalous transport processes, the transport exponent $\gamma$ can be defined as \cite{klages2008anomalous} 
\begin{equation}\label{eq.r} \gamma=\lim\limits_{t\rightarrow\infty}\frac{\ln\langle\Delta x^{2}_{t}\rangle}{\ln(t)}, \end{equation} 
where $\langle\Delta x^{2}_{t}\rangle$ represents the mean square displacement (MSD) at time $t$. The transport exponent $\gamma$ describes the transport regime of the system. For instance, if $\gamma < 1$, they are called sub-diffusive; when $\gamma=1$, they are known as normal diffusive and super-diffusive if $\gamma > 1$. Therefore, Eq.~\eqref{eq.r} holds $\gamma > 0$, and the generalized diffusion coefficient, denoted by $D_{\gamma}$ yields:
\begin{equation} D_{\gamma}=\lim_{t_\rightarrow\infty}\frac{\langle\Delta x^{2}_{t}\rangle}{t^{\gamma}} \in (0,\infty). \end{equation} 
Eq.~\eqref{eq.r} shows the operation of the transport process works in different systems. For example, systems with linear growth of MSD over time show normal diffusion \cite{morgado2002relation}. Conversely, systems that deviate from the normal transport exhibit anomalous diffusion. This occurs when the MSD grows non-linearly with time, resulting in unbounded variance \cite{klages2008anomalous, zaslavsky2002chaos}. 
%It has been seen in a lot of different types of physical systems, including the movement of molecules in living cells \cite{saxton2001anomalous}, disorder in solid-state systems \cite{cheruzel2003structures}; the dynamics of cell membranes \cite{day2009tracking}; the low-dimensional transport of heat systems \cite{lepri2003thermal}; the transport of soil \cite{martin2012physical}; Brownian motion \cite{morters2010brownian}; macroscopic diffusion models \cite{BRP24}; authors of \cite{GSSPCM18} explores the crossover from anomalous to normal diffusion using truncated power-law noise correlations, with applications to lipid bilayer dynamics; and many more. Plenty of work has been devoted to investigating the microscopic properties of such dynamics and their applications  \cite{klages2008anomalous, jepps2003wall, klages2007microscopic, froemberg2015asymptotic, MG08,BRP24,Tthesis18}.
It has been observed in many different types of physical systems, including the movement of molecules in living cells \cite{saxton2001anomalous}, disorder in solid-state systems \cite{cheruzel2003structures}, the dynamics of cell membranes \cite{day2009tracking}, the low-dimensional transport of heat systems \cite{lepri2003thermal}, the transport of soil \cite{martin2012physical}, Brownian motion \cite{morters2010brownian}, and macroscopic diffusion models \cite{BRP24}, authors of \cite{GSSPCM18} studied the crossover from anomalous to normal diffusion using truncated power-law noise correlations, with applications to lipid bilayer dynamics. Many studies have investigated the microscopic properties of such dynamics and their applications \cite{klages2008anomalous, SMKM24, jepps2003wall, klages2007microscopic,RMDG24,
froemberg2015asymptotic,MG08,BRP24,
Tthesis18,LPW24,TGPMS24}. Above all, Mu\~noz-Gil et al. \cite{Metal21} reported that machine-learning methods are better than traditional methods for analyzing anomalous diffusion, including estimating exponents, classifying models, and segmenting trajectories.

The study of anomalous transport and investigations of microscopic properties have been carried out using principles from statistical mechanics. For instance, the probability distribution function of displacement (PDF) provides the moments of the transport processes. However, many systems with vanishing Lyapunov exponents have become extremely complex and may lack PDF. 
%In such cases, the study of transport becomes highly complicated, and the microscopic properties of these systems can be determined through the ensemble time average of the positions 
Characterizing microscopic correlations is crucial to distinguish different transport mechanisms \cite{salari2015simple, giberti2019equivalence, vollmer2021displacement,T24}. These microscopic properties include position moments and correlations, where the moments are described by the power law as
\begin{equation}\label{eq.dis}  
\langle|\Delta x_t|^{j}\rangle \sim t^{\gamma(j)},\end{equation}
where $\langle \cdot \rangle$ denotes the ensemble average of the trajectories, the variable $j$ refers to the order of moments being considered, and $\gamma(j)$ is known as the transport exponent, which characterizes the spectrum of the moments of displacement, which is commonly referred to as the nature of the transport process. For $j=2$, exponent $\gamma(2)=\eta$, Eq.~\eqref{eq.dis} represents the power-law behavior of the MSD. As $\eta$ changes, different transport regimes appear. For example, when $0<\eta<1$, the transport is sub-diffusive and normal when $\eta=1$; when $1<\eta<2$, it is super-diffusive and ballistic when $\eta=2$ \cite{salari2015simple, giberti2019equivalence, vollmer2021displacement, klages2008anomalous}. The case \(\eta = 0\) does not necessarily imply the absence of transport; rather, it may indicate rapid transport, such as a logarithmic growth in MSD over time \cite{salari2015simple}.
A scale-invariant transport process occurs when the anomalous behavior of the MSD, which differs from the Gaussian model with linearly growing variance, extends uniformly to all moments of displacement as a linear function of the moment order \(j\), given by \(\gamma(j) = \kappa j\). If \(\gamma(j)\) varies non-linearly with \(j\), the scale invariance breaks down and the transport is classified as strongly anomalous \cite{vollmer2021displacement, CMMGV99}. The most commonly studied strongly anomalous scenario in the literature involves two distinct regimes, where the exponent of \(t\) is governed by two different linear expressions in terms of moment order
\begin{eqnarray}\label{eq:scale-invrt-momnts}
\gamma(j)=\left\{
\begin{array}{lrl}
\kappa j, & \text{for} & \quad j\leq j_c , \\[2mm]
j-(1-j)j_c, &\text{for} & \quad j>j_c,
\end{array}
\right.
\end{eqnarray}
where $\kappa$ and $j_c$ are given parameters.

Analogously, $j^{\text{th}}$ order PACF scale as the power law, which yields
\begin{equation}\label{eq.corr}
\langle\Delta x_{t_{1}}\Delta x_{t_{2}}\cdots\Delta x_{t_{j}}\rangle \sim t_j^{\gamma^{\prime}(j)}, \qquad t_1\leq t_2 \leq  \cdots t_j\,,
\end{equation}
here $\langle \cdot \rangle$ denotes the ensemble average over the trajectories at different times, and $j$ and $\gamma^{\prime}(j)$ is the order and transport exponent of the PACF, respectively.
% The anomalous diffusion has also be found in many stochastic systems and deterministic process~\cite{andersen2000simple,zaslavsky2002chaos,sanders2006occurrence,klages2007microscopic,sokolov2012models,salari2015simple}, mixed space of  Hamiltonian systems~\cite{castiglione1999strong}, Polygonal Billiards~\cite{sanders2006occurrence,jepps2006thermodynamics}, infinite Horizon Billiards~\cite{armstead2003anomalous,schmiedeberg2006superdiffusion,courbage2008problem}, maps with one dimension~\cite{pikovsky1991statistical,salari2015simple}, Inhomogeneous quenched stochastic models~\cite{burioni2010levy,bernabo2014anomalous}, Laser coded atoms diffusion~\cite{aghion2017large}, moving particles in living cancer cells~\cite{gal2010experimental}, dynamics of membranes~\cite{sneppen1994multidiffusion}, and diffusion molecules in Lipid bilayers~\cite{krapf2016strange}.\\

In recent years, the generalized central limit theorem (CLT) and non-normalizable densities have been used to explain the dynamics of some anomalous processes \cite{froemberg2015asymptotic,rebenshtok2014non,rebenshtok2014infinite}, and shed light on the relationship between deterministic dynamics and transport processes. Although microscopic properties such as position moments and correlations differentiate these processes, chaotic systems are characterized by rapid correlation decay and linear moment exponents, which are associated with standard diffusion. In contrast, non-chaotic systems often decay much slower than chaotic systems and exhibit non-linear diffusion, also known as anomalous diffusion \cite{klages2008anomalous,zaslavsky2002chaos,
klages2007microscopic,
salari2015simple,SMT12,jepps2006thermodynamics}. However, stochastic processes give rise to both normal and anomalous diffusion, depending on their rate of correlation decay, which raises
%and this duality opens 
various questions \cite{dettmann2000microscopic,cecconi2003origin,klages2008anomalous,zaslavsky2002chaos,klages2007microscopic,salari2015simple,denisov2003dynamical,li2005anomalous,barkai1999stochastic,barkai2000one}. In a specific realm, correlation functions are used to differentiate the transport properties of systems with different microscopic structures, but they yield the same position moments \cite{sokolov2012models,rebenshtok2014non,baule2007fractional,barkai2007multi,zaburdaev2008microscopic}. In the following, we discuss the deterministic and stochastic processes known as the SM and the LLg.

The SM \cite{salari2015simple} is a simple one-dimensional, non-chaotic, and deterministic dynamics with a vanishing Lyapunov exponent, which exhibits all possible diffusive regimes by tuning a single parameter $\alpha$. In the SM, trajectories move ballistically in their initial transit until they enter the stable cycle, \ie the $2$-period cycle, and then start oscillating towards their neighboring cells. Salari et al. \cite{salari2015simple} showed that the position moments of the SM have the same asymptotic scaling as those of the LLg when the parameters are appropriately tuned (\cf~Eq.~\eqref{eq.39}).
%Salari et. al. \cite{salari2015simple} analytically derived the position moments and determined the equivalence to the moments of the LLg under the functional relation between these dynamics. 
%This opens a window for a new of analysis of stochastic processes using deterministic dynamics, for which direct investigation of such systems is not possible. 
This framework allows us to study otherwise intractable stochastic processes using analytically solvable deterministic dynamics. The LLg is a random walk in quenched random environments, where scatterers are randomly distributed according to the L\'evy type probability density along a one-dimensional line with a non-equilibrium initial condition. This type of setup has been studied by several authors \cite{ben2000diffusion, burioni2010levy,burioni2010levyb}. Burioni et. al. \cite{burioni2010levy} calculated the position moments of the LLg by assuming a L\'evy type setup and verified the results with numerical simulations. Recently, Bianchi et al. \cite{BCLL16} presented a thorough mathematical framework to show a CLT, where Zamparo \cite{Z23} studied the large fluctuations and transport features of the LLg.

However, the agreement in position moments alone is insufficient to capture the microscopic dynamics of the system or to distinguish between different underlying classes of stochastic processes.
%However, the analytic treatment of higher-order PACFs (beyond $2$-point) in simple deterministic models remains largely undeveloped, and quantitative comparisons with the LLg in multiple time variables have not been fully explored.
%However, this does not suffice for the microscopic behavior of the system and indistinguishability of the stochastic process class. 
With this consistency, 
%Giberti et al. \cite{giberti2019equivalence} analytically derived the $2$-point PACF of the SM with some possible scalings and compared them with the numerically estimated power law behavior of the correlation functions of the LLg; interestingly, they found a remarkable agreement in both the normal and superdiffusive regimes. 
Giberti et al. \cite{giberti2019equivalence} analytically derived the low-order %$2$-point
PACF of the SM and found that it matches the asymptotic decay of the corresponding PACF in the LLg, even though one system is deterministic and the other is random.
Moreover, Vollmer et. al. \cite{vollmer2021displacement} investigate strong anomalous diffusion through the ``fly-and-die'' (FnD) dynamics, which offers an exact analytical solution for systems with ballistic excursions (or light fronts). This study presents a scaling form for the PACF, revealing a universal behavior across various dynamics such as the SM \cite{salari2015simple}, LLg \cite{barkai1999stochastic,barkai2000one, burioni2010levyb}, Lorentz gas \cite{BS80, D14}, L\'evy walks \cite{froemberg2013time,froemberg2013random}, and polygonal billiards \cite{sanders2006occurrence,jepps2006thermodynamics}. The FnD dynamics scaling form effectively predicts the time-dependent behavior of these systems, achieving robust data collapse. In addition, it addresses corrections to scaling, emphasizing their system-specific nature. Recently, Tayyab \cite{T24} evaluated the correlation function in a single scaling regime (\ie the ratio of times) and demonstrated that many systems exhibit universal behavior, even if they have different microscopic dynamics. However, the results yield the same moments and correlation functions;
%These findings show that the SM and the FnD effectively capture the essential features of the LLg correlations, particularly in the case of strong superdiffusive behavior. 
hence, these findings show that the dynamics of SM and FnD mimic the fundamental correlation characteristics of the LLg under significantly superdiffusive regimes. This alignment holds under specific parameter conditions but diminishes as the system transitions to normal diffusion.
This suggests that these dynamics mimic the anomalous transport properties of the LLg under certain conditions. 
%Therefore, low-order PACFs are insufficient for a comprehensive understanding of the transport properties of such systems and do not address the microscopic details of each class of stochastic processes.
However, the analytic treatment of higher-order PACFs (beyond $2$-point) in simple deterministic models remains largely undeveloped, and quantitative comparisons with the LLg in multiple time variables have not been fully explored.
Therefore, higher-order PACFs are essential for gaining insights into the transport dynamics of the LLg.

In this paper we go beyond previous results by deriving exact higher-order PACFs for the SM and explicitly comparing them to the LLg. Our key results are as follows:
\begin{itemize}
	\item 
		We derive the generalized ($j$-point) PACF of the SM in closed form, obtaining asymptotic scaling laws as functions of the map parameter $\alpha$. In particular, we obtain explicit analytic expressions for the $3$-point PACF $\phi_\alpha(m_1,m_2,m_3)$ and characterize its power-law scaling under different time-separation regimes.
	\item
		We compute the $3$-point PACF of the LLg numerically under the same time compositions as in the SM, and demonstrate that it exhibits power-law tails with exponents matching those of the SM. By matching the parameters and time scale regimes used in the SM, we ensure a meaningful comparison of the correlation decay between the two models.
		\item
			We find a remarkable agreement between the SM and LLg; in each temporal configuration studied, the LLg numerically estimated $3$-point correlations coincide with the analytic scaling derived for the SM. This strong overlap holds in the superdiffusive regime and confirms that the SM reproduces the key multi-time correlation structure of the LLg.
\end{itemize}
These results establish the SM as a simple deterministic dynamics that captures not only the MSD and $2$-point PACFs of a complex stochastic system, but also its higher-order correlation behavior. In doing so, we extend the state of the art in anomalous transport. Unlike prior work that focused on low-order correlations or universal scaling, our approach yields exact analytic formulas for higher-order PACFs in a model setting. This allows us to test universality in greater detail and to attribute stochastic correlation features to the underlying deterministic dynamics.

%In this paper, we derive the generalized (or $j$-point) PACF of the SM and its corresponding generalized scaling forms. This generalized approach provides a basis for testing the equivalence of the higher-order PACF between the SM and LLg, and can be used to determine the equivalence of the transport properties of these two processes in any order. Specifically, we derive the $3$-point PACF of SM as $\phi_\alpha(m_1,m_2,m_3)$ and their possible scalings, which contain the different possible relationships between the three times $m_1, m_2$ and $m_3$. To determine whether the statistical analysis distinguishes the transport properties of the LLg, we formulate the $3$-point position autocorrelations of the LLg using the same time composition as in the SM. Then, we compare the scalings of the $3$-point PACF of the SM with those of the numerically estimated LLg correlations. We find remarkable agreement between the numerically estimated correlation results for the LLg and those of the SM $3$-point PACF. 
%However, these derived $3$-point PACF alone are insufficient to establish full equivalence between the models; therefore, we modify our analysis to include higher-order PACF to further assess the equivalence of these models.

The remainder of this paper is organized as follows: Section \ref{sec.1} reviews the dynamics of the SM and its properties. In section \ref{subsec.3}, we derive the generalized (or $j$-point) PACF and some of its possible scaling. Section \ref{sec:3pointcorr_SM} is devoted to the $3$-point PACF and its scaling, which are presented in several lemmas below. Section \ref{sec.2} demonstrates the LLg, which characterizes the properties of the system, and presents the numerical results of the scaling applicability of the $3$-point PACF of the SM. Finally, in Sec. \ref{sec.5}, we conclude with our discussion.
%we generalized the expression for the PACF and find the asymptotic scalings for that. After that for the comparison, in Sec.~\ref{sec.2} we compare the LLg with the SM. For this, we derived the three-point PACF of the LLg numerically and compare the results with the SM. At the last Sec.~\ref{sec.5}, we conclude our paper with the conclusion.
\section{The Dynamics of the Slicer map}\label{sec.1}
Discrete-time deterministic dynamics called, the SM \(S_{\alpha}(x, \mathcal{M})\) was introduced by Salari et al. \cite{salari2015simple} on the unit interval \([0,1]\). The time evolution is given by a map 
\begin{equation*} S_\alpha: [0,1]\times \mathbb{Z} \rightarrow [0,1]\times \mathbb{Z}, \end{equation*} 
defined by 
\begin{subequations}\label{eq:SMdynamics} 
%\begin{eqnarray}\label{eq.1} 
%S_{\alpha}(x,m)=\left\{ \begin{array}{ccc} (x,m-1) & \hbox{if} & \hbox{$0\leq x < l_{m}$ \; \text{or} \; $1/2 < x \leq 1-l_{m}$}, \\[2mm] (x,m+1)& \hbox{if}&\hbox{$l_{m} \leq x \leq 1/2 $ \; \text{or} \; $1-l_{m} < x \leq 1$}. 
%\end{array} \right. 
%\end{eqnarray}
\begin{align}\label{eq.1}
x_{n+1} 
 =
 S_\alpha(x_n) = \left\{
    \begin{array}{rl}
      (x_n,\mathcal{M}-1) & \mbox{ if \quad} 0\leq x_n \le  \ell_{\mathcal{M}} \mbox{ \;or\; } \frac{1}{2} < x_n \leq 1-\ell_{\mathcal{M}},\\[3mm]
      (x_n,\mathcal{M}+1) & \mbox{ if \quad} \ell_{\mathcal{M}} < x_n \leq \frac{1}{2} \mbox{ \;or\; } 1-\ell_{\mathcal{M}} < x_n \leq 1 \, ,
    \end{array}
                \right.
\end{align}
where 
\( x_{n} = \{ x + n \} \), with \( x \) is the fractional part (\ie~\( 0 \leq x < 1 \)) and \( n \in \mathbb{N}_0 \) is a non-negative integer. In each term \( x_{n} \), \( n \) is added to \( x \), and the use of the fractional part function ensures that \( x_{n} \) remains within the interval \([0,1]\). The initial ensemble $x_0$ is chosen uniformly in the interval $[0,1]$.
For all integers $\mathcal{M}$, the family of \textit{slicers} $l_{\mathcal{M}}(\alpha)$ is defined as \begin{equation}\label{eq.len} l_{\mathcal{M}}(\alpha)=\frac{1}{(|\mathcal{M}|+2^{1/\alpha})^{\alpha}}, \quad \text{with}\quad \mathcal{M}\in\mathbb{Z}, \; \alpha>0. \end{equation} 
\end{subequations} 
%%%%%%%%%%%%%%%%%%%%%%%%%%%%%%%%%%%%%%%%%%%%%%%%
The dynamics of the SM is explained by Eq.~\eqref{eq:SMdynamics}: when a point \( x_n \) lies within the interval \([1/2, 1]\), the slicers move one step forward to cell \( \mathcal{M}+1 \) until \( x_n > l_{\mathcal{M}} \). At this point, the trajectories become periodic with a period of two, oscillating back and forth between the cells \( \mathcal{M} \) and \( \mathcal{M}+1 \). Similarly, when \( x_n \) is in the interval \([0, 1/2]\), the slicers become periodic with the same period, oscillating between cells \( \mathcal{M}-1 \) and \( \mathcal{M} \). The unit cell of the SM, \([0,1]\), is divided into two halves: \([0, 1/2]\) and \([1/2, 1]\), which represent the negative and positive halves of the chain, respectively. The two halves are symmetric and exhibit the same dynamics. Points within the interval \([1/2, 1]\) never reach cells indexed by \( \mathcal{M}-1 \), and points within the interval \([0, 1/2]\) never reach cells indexed by \( \mathcal{M}+1 \). To analyze the transport properties of the SM, we discuss the position moments in Lemma \ref{lem:SM-p-momnts}, and the PACF in the following sections.
\begin{lemma}\label{lem:SM-p-momnts} 
For $\alpha > 0$, the $j^{th}$ position moments of the SM for the uniformly distributed initial condition $x\in[0,1]$, scale asymptotically as \begin{eqnarray}\label{eq.10} 
\langle |\Delta x_m|^{j}\rangle\sim\left\{ \begin{array}{lll} \frac{2\,j}{j-\alpha}m^{j-\alpha}\,, & \mbox{for} & \hbox{$0<\alpha< j$}, \\ [2mm] 2\,j\ln(m)\,, & \mbox{for} &\hbox{$\alpha=j$},\\ [2mm] \text{const}\,, & \mbox{for} &\hbox{$\alpha>j$}.\\ \end{array} \right. \end{eqnarray} 
\end{lemma} 
\begin{proof} See \cite{salari2015simple, giberti2019equivalence, vollmer2021displacement}. \end{proof}
For the case when $j = 2$, Lemma \ref{lem:SM-p-momnts} provides asymptotic scaling for the MSD, which characterizes all possible diffusive regimes of transport by varying a single parameter, \(\alpha\). Specifically, the SM exhibits normal diffusion when \(\alpha = 1\). In other words, the MSD increases linearly with time between a pair of points in this case, describing the same phenomena as classical Brownian motion. If \( 0 < \alpha < 1 \), then the SM exhibits super-diffusive behavior, \ie the MSD grows even faster than linearly, which implies that the particles spread at a higher rate
than in normal diffusion. For \(1 < \alpha < 2 \), the behavior is sub-diffusive because the MSD growth is not linear. This means that the particles are not spreading as much, possibly because they are trapped or held in place by other effects. In the last case, \(\alpha = 2\), exhibits logarithmic MSD growth, which heuristically suggests that the particles do not spread over time.
\paragraph{\textbf{Physical Interpretation of $\alpha$}:}
Parameter $\alpha$ governs the interplay between the transport efficiency and environmental constraints. For the SM:
\begin{itemize}
    \item \textbf{$0 < \alpha < 1$ (superdiffusion):} 
Characterizes systems with sparse obstacles or persistent motion, analogous to the L\'evy flights observed in intracellular transport \cite{LSBWS17} or fractal media \cite{YKY24}.     
    \item \textbf{$\alpha = 1$ (normal diffusion):} Describes unconstrained motion in homogeneous environments, typical of Brownian particles in simple fluids \cite{LSBWS17,YKY24}.
    \item \textbf{$1 < \alpha < 2$ (subdiffusion):} 
Emerges when transient trapping occurs, as observed in crowded cellular environments \cite{LSBWS17} or viscoelastic materials \cite{CM16}.       
    \item \textbf{$\alpha = 2$ (saturation):} 
Corresponds to strongly impeded dynamics where particle motion becomes restricted, observed in dense colloidal systems \cite{SLX24} or glassy states \cite{WMC22}, where MSD plateaus owing to restricted motion.    
       %Signifies extreme confinement, such as dense obstacles e.g., glassy states \cite{} or rigid porous matrices \cite{}, 
\end{itemize}
%This parameter maps to physical observables in stochastic systems for instance, the scatterer density $\beta$ in the LLg; see Eq.~\eqref{eq.39} and quantifies the environmental complexity in real-world applications \cite{LSBWS17,YKY24}.
This parameter $\alpha$ maps to the scatterer density $\beta$ in the LLg (\cf~Eq.~\eqref{eq.39}) quantifying environmental complexity \cite{LSBWS17,YKY24}.
%%%%%%%%%%%%%%%%%%%%%%%%%%%%%%%%%%%%%%%%%%%%%%%%
\subsection{Generalized position auto-correlation function of the Slicer Map}\label{subsec.3}
The generalized (or \( j \)-point) PACF of the SM is denoted as \(\phi_\alpha(m_{1}, m_{2}, \ldots, m_{j})\), this function quantifies the distance traveled by trajectories at specified time indices \( m_1 \leq m_2 \leq \cdots \leq m_j \). It is defined as
\begin{align}
\phi_\alpha(m_{1}, m_{2}, \ldots, m_{j}) := 
\langle \Delta x_{m_{1}} \Delta x_{m_{2}} \cdots \Delta x_{m_{j}} \rangle = \langle (x_{m_{1}} - x_{0})(x_{m_{2}} - x_{0}) \cdots (x_{m_{j}} - x_{0}) \rangle,
\end{align}
where \(\langle \cdot \rangle\) represents the ensemble average over all trajectories, capturing the paths taken by the system over time. From the dynamics of the SM, we introduce the term \( l_{\mathcal{M}}^{+}(\alpha) \), known as the \textit{slicer} which evolve in the positive half of the dynamics. The slicer characterizes the positive portion length parameters of the trajectories, and is defined as follows:
\begin{equation}\label{eq.len2}
l_{\mathcal{M}}^{+}(\alpha) = 1 - l_{\mathcal{M}}(\alpha)=1 - \frac{1}{(|{\mathcal{M}}| + 2^{1/\alpha})^{\alpha}}, \quad \text{for } {\mathcal{M}} \in \mathbb{Z}, \; \alpha > 0.
\end{equation}
This definition constrains \( l_{\mathcal{M}}^{+}(\alpha) \) to the interval \([1/2, 1)\) for each integer \( {\mathcal{M}} \). Consequently, for \( x \in [1/2, 1) \), we employ the condition
\begin{equation}\label{eq.xlim}
l_{{\mathcal{M}}-1}^{+}(\alpha) < x < l_{{\mathcal{M}}}^{+}(\alpha).
\end{equation}
This satisfy
\begin{equation}
\frac{1}{2} = l_{0}^{+}(\alpha) < l_{1}^{+}(\alpha) < \cdots < l_{{\mathcal{M}}-1}^{+}(\alpha) < l_{{\mathcal{M}}}^{+}(\alpha) = 1,\qquad \text{such as}\quad \lim\limits_{{\mathcal{M}}\rightarrow \infty} l_{{\mathcal{M}}}^{+}(\alpha) = 1\,.
\end{equation}
We denote the length of interval
\begin{align*}
\Delta_k (\alpha) = l_{k}^{+}(\alpha) - l_{k-1}^{+}(\alpha)\,,\quad \text{by construction this adds upto} \; 1/2,\; \ie \sum\limits_{k=1}^{\infty} \Delta_k(\alpha) = \frac{1}{2}\,,
\end{align*}
where
\begin{equation}
\Delta_k(\alpha) = \frac{\alpha}{k^{\alpha + 1}}\left(1+O(k^{-1})\right)\,.
\end{equation}
Next, we define the function \(\rho_{\alpha}^{m}(x)\), which represents the maximum distance that trajectories can travel at any point \( x \in [1/2, 1) \) is $k$
\begin{equation}\label{maxrho}
\rho_{\alpha}^{m}(x) = \min\{\rho_{\alpha}^{m}(x), m\} \rightarrow k, \quad \text{for } m \in \mathbb{Z}, \; \alpha > 0, \; x \in [1/2, 1).
\end{equation}
For a generalized function that involves the time indices \( j \), we consider the expression
\[
\min\{\rho_{\alpha}^{m_{1}}(x), m_{1}\}, \min\{\rho_{\alpha}^{m_{2}}(x), m_{2}\}, \ldots, \min\{\rho_{\alpha}^{m_{j}}(x), m_{j}\}.
\]
The average of the correlation function is then obtained by integrating over the interval \((1/2, 1]\). This integral accounts for the maximum distances traveled by the trajectories at these points
\begin{align}\label{eq.27a}
\phi_\alpha(m_{1}, m_{2}, \cdots, m_{j}) = 2 \int_{1/2}^{1} \min\{\rho_{\alpha}^{m_{1}}(x), m_{1}\} \, \min\{\rho_{\alpha}^{m_{2}}(x), m_{2}\} \cdots \min\{\rho_{\alpha}^{m_{j}}(x), m_{j}\} \, dx.
\end{align}
The interval \((1/2, 1]\) is subdivided into \( j+1 \) parts as follows
\[
(1/2,1] = Z_{1}^{m_{1}} \cup Z_{m_{1}+1}^{m_{2}} \cup Z_{m_{2}+1}^{m_{3}} \cup \cdots \cup Z_{m_{j}+1}^{\infty},
\]
where each \( Z \)-region defines the regimes in which the trajectories oscillate. Specifically, these subintervals are defined as
\begin{equation}
\begin{aligned}\label{eq:itrvl-divs}
Z^{m_{1}}_{1} &= \{x\in [1/2,1) : 0<\rho_{\alpha}^{m_{1}}\leq m_{1} \}\Rightarrow \prod_{i=1}^{j}min\{\rho_{\alpha}^{m_{i}}(x),m_{i}\}= \prod_{i=1}^{j}\rho_{\alpha}^{m_{i}},\\
Z_{m_{1}+1}^{m_{2}} &= \{x\in [1/2,1) : m_{1}<\rho_{\alpha}^{m_{2}}\leq m_{2} \}\Rightarrow \prod_{i=1}^{j}min\{\rho_{\alpha}^{m_{i}}(x),m_{i}\}= m_{1}\prod_{i=2}^{j} \rho_{\alpha}^{m_{i}}(x),\\
Z_{m_{2}+1}^{m_{3}}&=\{x\in [1/2,1) : m_{2}<\rho_{\alpha}^{m_{3}}\leq m_{3} \}\Rightarrow \prod_{i=1}^{j}min\{\rho_{\alpha}^{m_{i}}(x),m_{i}\}= m_{1}m_{2} \prod_{i=3}^{j}\rho_{\alpha}^{m_{i}}(x),
\\
\vdots\quad & \quad \vdots\qquad \vdots \qquad \vdots\qquad \quad \vdots \qquad \vdots\qquad \vdots \qquad \vdots\qquad \quad \vdots \qquad \vdots \qquad \vdots \\
Z^{\infty}_{m_{j}}&=\{x\in [1/2,1) : m_{j}<\rho_{\alpha}^{m_{j}} \}\Rightarrow \prod_{i=1}^{j}min\{\rho_{\alpha}^{m_{i}}(x),m_{i}\}= \prod_{i=1}^{j}m_{i}.
\end{aligned}
\end{equation}
%\begin{eqnarray}
%Z^{m_{1}}_{1} &=& \{x\in [1/2,1) : 0<\rho_{\alpha}^{m_{1}}\leq m_{1} \}\Rightarrow \prod_{i=1}^{j}min\{\rho_{\alpha}^{m_{i}}(x),m_{i}\}= \prod_{i=1}^{j}\rho_{\alpha}^{m_{i}}\\
%Z_{m_{1}+1}^{m_{2}} &=& \{x\in [1/2,1) : m_{1}<\rho_{\alpha}^{m_{2}}\leq m_{2} \}\Rightarrow \prod_{i=1}^{j}min\{\rho_{\alpha}^{m_{i}}(x),m_{i}\}= m_{1}\prod_{i=2}^{j} \rho_{\alpha}^{m_{i}}(x)\\
%%Z_{m_{2}+1}^{m_{3}}&=\{x\in [1/2,1) : m_{2}<\rho_{\alpha}^{m_{3}}\leq m_{3} \}\Rightarrow \prod_{i=1}^{j}min\{\rho_{\alpha}^{m_{i}}(x),m_{i}\}= m_{1}m_{2} \prod_{i=3}^{j}\rho_{\alpha}^{m_{i}}(x)\\
%%Z^{m_{4}}_{m_{3}+1}&=\{x\in [1/2,1) : m_{3}<\rho_{\alpha}^{m_{3}}<m_{4} \}\Rightarrow \prod_{i=1}^{j}min\{\rho_{\alpha}^{m_{i}}(x),m_{i}\}=m_{1}m_{2}m_{3}\prod_{i=2}^{j} \rho_{\alpha}^{m_{i}}(x)\\
%\vdots\quad && \quad \vdots\qquad \vdots \qquad \vdots\qquad \quad \vdots \qquad \vdots\qquad \vdots \qquad \vdots\qquad \quad \vdots \qquad \vdots\\
%Z^{\infty}_{m_{j}}&=&\{x\in [1/2,1) : m_{j}<\rho_{\alpha}^{m_{j}} \}\Rightarrow \prod_{i=1}^{j}min\{\rho_{\alpha}^{m_{i}}(x),m_{i}\}= \prod_{i=1}^{j}m_{i}.
%\end{eqnarray}
In the first sub-interval, \( Z^{m_{1}}_{1} \), all trajectories oscillate within their neighboring cells. In the second sub-interval, \( Z_{m_{1}+1}^{m_{2}} \), the trajectories at time \( m_{1} \) are traveling, whereas the others continue to oscillate within their neighboring cells. This process continues until the final subinterval \( Z^{\infty}_{m_{j}} \), where all trajectories are traveling. 

Thus, Eq.~\eqref{eq.27a} is subdivided into \( j+1 \) parts, according to the regimes defined in Eq.~\eqref{eq:itrvl-divs}, where $\rho_{\alpha}^{m_{i}}$ takes a constant value $k$ and can be expressed as
%In the first interval, say $Z^{m_{1}}_{1}$, all the trajectories are oscillating into their neighboring cells; for the second interval $Z_{m_{1}+1}^{m_{2}}$, shows the trajectories for time $m_{1}$ are travelling while other trajectories are oscillating among their neighbouring cells.
%in the third interval, $Z_{m_{2}+1}^{m_{3}}$ the trajectories are travelling at time $m_1$ and $m_2$ while the other trajectories are still oscillating. 
%Analogously this process will continue until the last interval, say $Z^{\infty}_{m_{3}+1}$ where all the trajectories are travelling. Thus, Eq.~\eqref{eq.27a} is subdivided into $(j+1)$ parts in the view of regimes defined above 
\begin{subequations}\label{eq:genrlzd-corr}
\begin{align}\label{eq:jpoint-corr-intgrl}
\phi_\alpha(m_{1},m_{2},\cdots,m_{j})\ \sim & 2\int\limits_{Z^{m_{1}}_{1}}\prod_{i=1}^{j}\rho_{\alpha}^{m_{i}}dx
+
2m_{1}\int\limits_{Z_{m_{1}+1}^{m_{2}}}\prod_{i=2}^{j} \rho_{\alpha}^{m_{i}}(x)dx
+
2m_{1}m_{2}\int\limits_{Z_{m_{2}+1}^{m_{3}}}\prod_{i=3}^{j}\rho_{\alpha}^{m_{i}}(x)dx\nonumber\\
&+\cdots + 2 \prod_{i=1}^{j}m_{i}\int\limits_{Z^{\infty}_{m_{j}}}dx.\\
\sim & 2\sum_{k=1}^{m_{1}}k^{j}\Delta_k(\alpha) 
+ 
2m_{1}\sum_{k=m_{1}+1}^{m_{2}}k^{j-1}
\Delta_k(\alpha) 
+ 2m_{1}m_{2} \sum_{k=m_{2}+1}^{m_{3}}k^{j-2}\Delta_k(\alpha)\nonumber \\
&+\cdots+
2\prod_{i=1}^{j}m_{i}\sum_{k=m_{j}+1}^{\infty}\Delta_k(\alpha).
\end{align}
\end{subequations}
%\begin{align}
%\phi(m_{1},m_{2},\cdots,m_{j})&\sim2\sum_{k=1}^{m_{1}}k^{j}\frac{\alpha}{k^{1+\alpha}}(1+O(k^{-1})) + 
%2m_{1}\sum_{k=m_{1}+1}^{m_{2}}k^{j-1}\frac{\alpha}{k^{1+\alpha}}(1+O(k^{-1})) +2m_{1}m_{2} 
%\nonumber \\
%&\times\sum_{k=m_{2}+1}^{m_{3}}k^{j-2}\frac{\alpha}{k^{1+\alpha}}(1+O(k^{-1}))
%%+2m_{1}m_{2}m_{3}\sum_{k=m_{3}+1}^{m_{4}}k^{j-3}
%%\frac{\alpha}{k^{1+\alpha}}(1+O(k^{-1}))\nonumber\\
%%+\cdots+2\prod_{i=1}^{j-1}m_{i}\sum_{k=m_{j-1}+1}%^{m_{j}}k\frac{\alpha}{k^{1+\alpha}}(1+O(k^{-1}))
%+\cdots+
%2\prod_{i=1}^{j}m_{i}\sum_{k=m_{j}+1}^{\infty}\frac{\alpha}{k^{1+\alpha}}(1+O(k^{-1})).
%\end{align}
%\textcolor{red}{As we defined as, the maximum distance travel by the  trajectories at different time is equal to $k$  and the length interval $\Delta_{k}(\alpha)$ can be find by using the equation of length of the Slicer given in Eq. \ref{eq.len}. Hence, by using the value of maximum distance travel by the trajectories and the length interval solve the above equation and it gives the generalize expression. That is given as}
Therefore, simple calculations allow us to write the generalized PACF as
\begin{equation}\label{eq.30}
\phi_\alpha(m_{1},m_{2},\cdots,m_{j}) \sim 2  \sum_{l=1}^{j} \prod_{i=0}^{j-l}m_{i} \left[\frac{\alpha\, m^{l-\alpha}_{(j+1)-l}}{(\alpha-l)\left((l-1)-\alpha\right)}\right] - \frac{\alpha}{j-\alpha}\,,
\end{equation}
where $j$ represents the order of the correlation function and $m_{0}=1$. This expression provides 
a new era for computing the PACF for any order. Consequently, to study the peculiar microscopic transport properties in Lemma \ref{lem:jpoint-scaling}, we determine the generalized asymptotic behavior for 
various time relationships in this generalized PACF expression, Eq.~\eqref{eq.30}.
\subsubsection{Asymptotic scaling for generalized position auto correlation function}\label{gen}
The asymptotic scaling of the generalized PACF depends on several parameters, such as the diffusion parameter $\alpha$ and relations between times. Here, we derive the asymptotic scaling for the generalized PACF and summarize it in Lemma \ref{lem:jpoint-scaling}.
%%%%%%%%%%%%%%%%%%%%%%%%%%%%%%%%%%%%%%%%%%%%%%%%%%%%%%%%%%%%%%%%%%%%%%%%%%%%%%%%%%%%%%%%%%%%%%%%%%%%
\begin{lemma}\label{lem:jpoint-scaling}
Given $\alpha > 0$, the $j$-point PACF $\phi_{\alpha}(m_1, m_2, \dots, m_j)$, defined in Eq.~\eqref{eq.30}, asymptotically scales as follows for $m_1 \leq m_2 \leq \cdots \leq m_j$ and $m_k\rightarrow\infty$ for all non-fixed times in $\{1,2,\cdots,j\}:$
\begin{subequations}
\begin{itemize}
\item $k\in\{0,1,2,\cdots,j-1\}$ denotes the number of fixed initial times (\ie $m_1,\dots,m_k$ held constant while $m_{k+1},\cdots,m_j \rightarrow \infty$)
\item Time lags $\tau_{l-1} \in m_l - m_{l-1}$ satisfy:
\begin{equation}\label{eq.32}
\tau_{l-1} \in 
\begin{cases} 
\sin m_1, & \text{(oscillatory separation)},\\
\ln m_1, & \text{(logarithmic growth)}, \\
m_1^q, \quad 0 < q < 2, & \text{(power-law growth)}, \\
\text{const.}, & \text{(fixed time separation)}.
\end{cases}
\end{equation}
\end{itemize}
The asymptotic scaling is given by:
\begin{equation}\label{eq.31}
\phi_{\alpha}(m_1,m_2,\dots,m_j) \sim 
\begin{cases} 
\frac{2j}{j-\alpha}m_{1}^{j-\alpha}, & 0<\alpha<j,\ q<1, \\[1.5ex]
\frac{2(j-1)}{(j-1)-\alpha}m_{1}^{1+q((j-1)-\alpha)}, & 0<\alpha<j-1,\ q>1, \\[1.5ex]
\frac{2\alpha}{(j-\alpha)(\alpha-(j-1))}m_{1}^{j-\alpha}, & j-1<\alpha<j,\ q>1, \\[2.5ex]
\frac{2(j-k)\prod_{i=1}^k m_i}{(j-k)-\alpha}m_{k+1}^{(j-k)-\alpha}, & 
   \substack{0<\alpha<j-k \\[0.5ex] k \in \{1,\dots,j-1\}}, \\[2.5ex]
\frac{2(j-1-k)\prod_{i=1}^k m_i}{(j-1-k)-\alpha}m_{k+1}^{1+q((j-1-k)-\alpha)}, & 
   \substack{0<\alpha<j-1-k,\ q>1 \\[0.5ex] k \in \{1,\dots,j-2\}}, \\[2.5ex]
\frac{2\alpha(j-k-\alpha)\prod_{i=1}^k m_i}{(j-k-\alpha)(\alpha-(j-1-k))}m_{k+1}^{j-k-\alpha}, & 
   \substack{j-1-k<\alpha<j-k,\ q>1 \\[0.5ex] k \in \{1,\dots,j-2\}}, \\[1.5ex]
j^\prime\ln m_1, & \alpha = j, \\[1ex]
\text{const.}, & \alpha > j,
\end{cases}
\end{equation}
\end{subequations}
where $j^\prime$ depends on $\alpha$ and $j$, and $m_k \to \infty$ for all non-fixed times.
\end{lemma}
\begin{proof} 
For $m_{k}\rightarrow\infty$ for each $k\in \{1,2,\cdots, j\}$, and the time lag $\tau_{l-1}$ represented in Eq. \eqref{eq.32}.
%$=m_{l}-m_{l-1}$, where $l\in[2,3,4,\cdots,j]$ given as
%\begin{eqnarray}\label{eq.32}
%\tau_{l-1}=\left\{
%\begin{array}{cc}
%\sin m_{1}\,, &\hbox{} \\ [2mm]
%\ln m_{1}\,, &\hbox{} \\[2mm]
%m_{1}^{q}\,,&\hbox{$0<q<2$}\,,\\ [2mm]
%\text{const}.&\hbox{}
%\end{array}
%\right.
%\end{eqnarray}
All terms in Eq.~\eqref{eq.30} significantly contributes to the power-law behavior when $0<\alpha<j$. Therefore, Eq.~\eqref{eq.30} yields as
\begin{align*}
\phi_\alpha(m_{1},m_{2},\cdots,m_{j})\sim  \frac{2j}{j-\alpha}m_{1}^{j-\alpha}, \quad\;\text{for}\;\;\quad 0<\alpha<j\,,\quad  q<1.
\end{align*}
For $m_{k}\rightarrow\infty$ where $k\in \{1,2,\cdots, j\}$ with a power-law-type time lag $\tau_{l-1}\in m_{1}^{q}$ and $q>1$. Then, all the terms from Eq.~\eqref{eq.30} contribute to the power law behavior as follows: where $0<\alpha<j-1$ all  terms contribute to the power law behavior except the first term, which is given as
\begin{align*}
\phi_\alpha(m_{1},m_{2},\cdots,m_{j})\sim\frac{2(j-1)}{(j-1)-\alpha}m_{1}^{1+q((j-1)-\alpha)}\,,\quad\;\;\text{for}\;\quad\;\;0<\alpha<j-1,\;\;q>1.
\end{align*} 
For $(j-1)<\alpha<j$, the first term in Eq.~\eqref{eq.30} contributes to the power-law behavior, which is given as
\begin{align*}
\phi_\alpha(m_{1},m_{2},\cdots,m_{j})\sim\frac{2\alpha}{(j-\alpha)(\alpha-(j-1))}m_{1}^{j-\alpha}\,,\quad\;\; \text{for}\qquad j-1<\alpha<j\,,\;\quad q>1.
\end{align*}
For the fixed $k$ defined in Lemma \ref{lem:jpoint-scaling},
%fixed times with interval  $\{1,2,3,\cdots,j-1\}$, 
the $m_{k+1}$ terms in Eq. \eqref{eq.30} contributes to the power law behavior over $0<\alpha<j-k$. From Eq. \eqref{eq.30} we have
\begin{align*}
\phi_\alpha(m_{1},m_{2},\cdots,m_{j})\sim \frac{2(j-k)\prod_{i=1}^{k}m_{i}}{((j-k)-\alpha)}m_{k+1}^{(j-k)-\alpha}&\,,\quad \text{for}\quad\;\; 0<\alpha<j-k.
\end{align*}
%where $k$ denotes the count of fixed initial times as defined in Lemma \ref{lem:jpoint-scaling}.
%\noindent As defined in Lemma \ref{lem:jpoint-scaling}, 
%$k$ denotes the count of fixed initial times,
%%For some $k$ fixed and comparatively small time having interval $\{1,2,3,\cdots,j-1\}$, 
%along with other times are running with power law type time interval $m_{k+1}^{q}$ having $q>1$. Therefore, the $m_{k+1}$ terms in Eq.~\eqref{eq.30} contributes to the power law behavior over $0<\alpha<j-1-k$ and $j-1-k<\alpha<j-k$. From Eq.~\eqref{eq.30} we have

\noindent As defined in Lemma~\ref{lem:jpoint-scaling}, $k$ denotes the number of fixed initial times, whereas the remaining times evolve with a power-law type time interval $m_{k+1}^{q}$, where $q > 1$. Therefore, the $m_{k+1}$ terms in Eq.~\eqref{eq.30} contribute to the power-law behavior over the intervals $0 < \alpha < j - 1 - k$ and $j - 1 - k < \alpha < j - k$. From Eq.~\eqref{eq.30}, we have
\begin{align*}
\phi_\alpha(m_{1},m_{2},\cdots,m_{j})\sim \frac{2(j-1-k)\prod_{i=1}^{k}m_{i}}{(j-1-k)-\alpha}m_{k+1}^{1+q((j-1-k)-\alpha)}&\,,\quad \text{for}\quad\;\; 0<\alpha<j-1-k.
\end{align*}
Analogous derivation for the case when $j-1-k<\alpha<j-k$, allow us to write
\begin{align*}
\phi_\alpha(m_{1},m_{2},\cdots,m_{j})\sim \frac{2\alpha(j-k-\alpha)\prod_{i=1}^{k}m_{i}}{(j-k-\alpha)(\alpha-(j-1-k))}m_{k+1}^{j-k-\alpha}&\,,\quad \text{for}\quad\;\; j-1-k<\alpha<j-k,
\end{align*}
%where $k$ is number of constant and small time.
where $k$ denotes the count of fixed initial times as defined in Lemma \ref{lem:jpoint-scaling}.

\noindent For $\alpha=j$, the first term in Eq. \eqref{eq:genrlzd-corr} diverges logarithmically as $\sim j^\prime \ln m_{1}$, where $j^\prime$ depends on $\alpha$ and $j$. For $\alpha>j$, all terms in Eq. \eqref{eq:genrlzd-corr}, yields a constant value. This completes the proof of Lemma \ref{lem:jpoint-scaling}.
\end{proof}
\subsection{$3$-point position autocorrelation function}\label{sec:3pointcorr_SM}
The $3$-point PACF can be obtained using Eq. \eqref{eq:genrlzd-corr} by requesting $j=3$, the correlation function for $m_1\leq m_2\leq m_3$ yields 
%Third we have $j=3$ here we have some times are fixed say $k$ is the number of fixed time, consider $k=2$, Eq.~\eqref{eq.32} is the consequence of Eq.~\eqref{eq.17}.
\begin{subequations}\label{eq:3point-corr-intgrl}
\begin{align}\label{intpart}
&\phi_\alpha(m_{1},m_{2},m_{3})\nonumber\\ 
& \sim 2\int\limits_{Z^{m_{1}}_{1}}\prod_{i=1}^{3}\rho_{\alpha}^{m_{i}}dx
+
2m_{1}\int\limits_{Z_{m_{1}+1}^{m_{2}}}\prod_{i=2}^{3} \rho_{\alpha}^{m_{i}}(x)dx
+
2m_{1}m_{2}\int\limits_{Z_{m_{2}+1}^{m_{3}}}\rho_{\alpha}^{m_{3}}(x)dx
+ 2 \prod_{i=1}^{3}m_{i}\int\limits_{Z^{\infty}_{m_{3}}}dx\,,\\[2mm]
&\sim 2\sum_{k=1}^{m_{1}}k^{3}\Delta_k(\alpha)+2m_{1}\sum_{k=m_{1}+1}^{m_{2}}k^{2}\Delta_k(\alpha)+2\prod_{i=1}^{2}m_{i} \sum_{k=m_{2}+1}^{m_{3}}k\Delta_k(\alpha)+2\prod_{i=1}^{3}m_{i}\sum_{k=m_{3}+1}^{\infty}\Delta_k(\alpha)\label{eq:3point-corr-sum}.
\end{align}
\end{subequations}
%By performing 
Simple calculations on Eq. \eqref{eq:3point-corr-intgrl} by using the motion of the trajectories given in Eq.~\eqref{eq:itrvl-divs} 
yield the
%allows us to write the 
$3$-point PACF in simplified form as
\begin{equation}\label{eq.16}
\phi_\alpha(m_{1},m_{2},m_{3}) \sim\frac{2\alpha \;m_{1}^{3-\alpha}}{(\alpha-3)(2-\alpha)} + \frac{2\alpha\;m_{1}m_{2}^{2-\alpha}}{(\alpha-2)(1-\alpha)}+\frac{2\;m_{1}m_{2}m_{3}^{1-\alpha}}{1-\alpha}
+\frac{2\alpha}{\alpha-3}.
\end{equation}
Now, we define the relationships between the three times $m_1$, $m_2$, and $m_3$ before deriving their asymptotic scaling behavior in the section \ref{sec:corr-scalings}.
%The asymptotic behaviour depends on the $\alpha$ and relation between $m_{3},m_{2},m_{1}$ times. Here, we discuss the seven different relations given below,
%\begin{description}
%\item[(i)]
%	$m_{1}, m_{2}$ are fixed and set $m_{3}\rightarrow\infty$
%\item[(ii)]
%	$m_{1}$ is fixed, $m_{2},m_{3}\rightarrow\infty$ and $m_{3}=m_{2}+h$
%\item[(iii)]	
%	$m_{1},m_{2},m_{3}\rightarrow\infty$, $m_{2}=m_{1}+h_{1}$ and $m_{3}=m_{1}+h_{1}+h_{2}$
%\item[(iv)]	
%	$m_{1},m_{2},m_{3}\rightarrow\infty$, $m_{2}=m_{1}+m_{1}^{q}$ and $m_{3}=m_{1}+m_{1}^{q}+m_{1}^{q}$, where $q$ is a positive number
%\item[(v)]	
%	$m_{1},m_{2},m_{3}\rightarrow\infty$, $m_{2}=m_{1}+\ln(m_{1})$ and $m_{3}=m_{1}+\ln(m_{1})+\ln(m_{1})$
%\item[(vi)]
%  $m_{1}m_{2},m_{3}\rightarrow\infty$, $m_{2}=m_{1}+\sin(m_{1})$ and $m_{3}=m_{1}+\sin(m_{1})+\sin(m_{1})$.
%\item[(vii)]  
%    $m_{1},m_{2},m_{3}\rightarrow\infty$, $m_{2}=m_{1}+h$ and $m_{3}=m_{1}+h+m_{1}^{q}$, where $q$ is a positive number and $h$ is a constant.
%\item[(viii)]    
%     $m_{1}$ is fixed, $m_{2},m_{3}\rightarrow\infty$ and  $m_{3}=m_{2}+m_{2}^{q}$
%\end{description}
\begin{description}
\item[i.]
	$m_{1}$ and $m_{2}$ are fixed as $m_{3} \rightarrow \infty$.
\item[ii.]
	$m_{1}$ is fixed; $m_{2}, m_{3} \rightarrow \infty$, and $m_{3} = m_{2} + \tau$.
\item[iii.]	
	$m_{1}, m_{2}, m_{3} \rightarrow \infty$, $m_{2} = m_{1} + \tau$, and $m_{3} = m_{1} + \tau + \xi$.
\item[iv.]	
	$m_{1}, m_{2}, m_{3} \rightarrow \infty$, $m_{2} = m_{1} + m_{1}^{q}$, and $m_{3} = m_{1} + m_{1}^{q} + m_{1}^{q}$.
	%, where $q$ is a positive number.
\item[v.]	
	$m_{1}, m_{2}, m_{3} \rightarrow \infty$, $m_{2} = m_{1} + \ln(m_{1})$, and $m_{3} = m_{1} + \ln(m_{1}) + \ln(m_{1})$.
\item[vi.]
	$m_{1}, m_{2}, m_{3} \rightarrow \infty$, $m_{2} = m_{1} + \sin(m_{1})$, and $m_{3} = m_{1} + \sin(m_{1}) + \sin(m_{1})$.
\item[vii.]  
    $m_{1}, m_{2}, m_{3} \rightarrow \infty$, $m_{2} = m_{1} + \tau$, and $m_{3} = m_{1} + \tau + m_{1}^{q}$.
    %, where $\tau$ is a constant and $q$ is a positive number.
\item[viii.]    
    $m_{1}$ is fixed; $m_{2}, m_{3} \rightarrow \infty$, and $m_{3} = m_{2} + m_{2}^{q}$.
\end{description}
Here, $\tau$ and $\xi$ are constants, and $q$ is a positive number.
\subsection{Scalings of the position autocorrelation function}\label{sec:corr-scalings}
We derive several scaling forms of the $3$-point PACF, considering different relationships between the three times \( m_1 \), \( m_2 \), and \( m_3 \) for \( m_1 \leq m_2 \leq m_3 \). We observe distinct power-law scaling for different choices of parameter \(\alpha\) and time relationships.
\subsubsection{Scaling of $\phi(m_{1},m_{2},m_{3})$ with $m_{1},m_{2}$ fixed for $m_{3}\rightarrow\infty$}
The asymptotic behavior of the 3-point PACF \(\phi_\alpha(m_1, m_2, m_3)\) for \(m_1 \leq m_2 \leq m_3\), with \(m_1\) and \(m_2\) fixed as \(m_3 \rightarrow \infty\), is summarized in Lemma \ref{lem:m1m2_fixed}.
\begin{lemma}\label{lem:m1m2_fixed}
For $\alpha>0$, and fixed two times \(m_1\) and 
\(m_2\), with \(m_3 \rightarrow \infty\) the PACF 
\(\phi_\alpha\) asymptotically behaves as
\begin{eqnarray}\label{eq.17}
\phi_\alpha(m_{1},m_{2},m_{3})\sim\left\{
\begin{array}{lll}
\frac{2m_{1}m_{2}}{1-\alpha}m_{3}^{1-\alpha}\,, & 0<\alpha< 1\,, \\[2mm]
2m_{1}m_{2}\ln(m_{3})\,,&\alpha=1\,,\\[2mm]
\text{const}\,,&\alpha>1.	
\end{array}
\right.
\end{eqnarray}
\end{lemma}
\begin{proof}
The first two terms in Eq. \eqref{eq.16} has no  significant contribution to the power law behavior, only the third term with $m_{3}^{1-\alpha}$ contributes to the asymptotic behavior over $(0<\alpha<1)$. From Eq. \eqref{eq.16} we have
\begin{eqnarray}
\phi_\alpha(m_{1},m_{2},m_{3})\sim \frac{2m_{1}m_{2}}{1-\alpha} m_{3}^{1-\alpha}\,, \qquad \text{for}\quad 0 <\alpha <1. \nonumber
\end{eqnarray}
For $\alpha=1$ the first two terms in Eq. \eqref{eq:3point-corr-intgrl} takes a constant value, whereas the third term diverges logarithmically to $
%\phi(m_{3},m_{2},m_{1})
\sim 2m_{1}m_{2}\ln m_{3}$. For $\alpha>1$, all the terms in Eq. \eqref{eq:3point-corr-intgrl}, yields a constant value. This completes the proof of Lemma \ref{lem:m1m2_fixed}.
%These are the analytical proof of the lemma.2, this lemma shows that when one time is running the auto-correlation of the SM gives the exponent say $w_{1}$, which is equal to $1-\alpha$.
\end{proof}
%%%%%%%%%%%%%%%%%%%%%%%%%%%%%%%%%%%%%%%%%%%%%%%%%%
\subsubsection
%{Scaling of $\phi_\alpha(m_{1},m_{2},m_{3})$ with $m_{1}$ fixed for $m_{2},m_{3}\rightarrow\infty$ and $m_{3}=m_{2}+\xi$, $\xi>0$ constant}
{Scaling of \(\phi_{\alpha}(m_{1},m_{2},m_{3})\) for fixed \(m_{1}\) as \(m_{2},\,m_{3}\to\infty\) with \(m_{3}=m_{2}+\xi\) (\(\xi>0\) fixed)}
\label{sec:m1fixd_m2_run_SM}
The asymptotic behavior of the $3$-point PACF $\phi_\alpha(m_{1},m_{2},m_{3})$ for $ m_{1}\leq m_{2}\leq m_{3}$, with one fixed time $m_{1}$ and
$m_{2},m_{3}\rightarrow\infty$ having a fixed splitting time $m_{3}-m_{2}=\xi=\text{constant}$, where $\xi>0$. The leading-order asymptotic behavior is summarized in Lemma \ref{lem:m1fixd_t1t2run}.
\begin{lemma}\label{lem:m1fixd_t1t2run}
For $m_{i}\rightarrow\infty$ for each $i\in\{2,3\}$, with $m_1$ fixed and constant time lag $m_{3}-m_{2}=\xi$, PACF $\phi_\alpha(m_{1},m_{2},m_{2})$ asymptotically behaves as
\begin{eqnarray}\label{eq.18}
\phi_\alpha(m_1,m_2,m_{2}+\xi)\sim \left\{
\begin{array}{lll}
\frac{4m_{1}}{2-\alpha}m_{2}^{2-\alpha}\,, &\hbox{$0<\alpha< 2$}\,, \\ [2mm]
4m_{1}\ln(m_{2})\,,&\hbox{$\alpha=2$}\,,\\ [2mm]
\text{const}\,,&\hbox{$\alpha>2$}.	
\end{array}
\right.
\end{eqnarray}
\end{lemma}
\begin{proof}
For fixed time $m_1$, the first term in Eq. \eqref{eq.16} does not significantly contribute to the power-law behavior, whereas the second and third terms determine the asymptotic behavior for \(0 < \alpha < 2\). Therefore, Eq.~\eqref{eq.16} results in the following
\begin{eqnarray}
\phi_\alpha(m_1,m_{2},m_{3})\sim \frac{4m_{1}}{2-\alpha} m_{2}^{2-\alpha}, && 0<\alpha<2.\nonumber
\end{eqnarray}
For $\alpha=2$ the first term in Eqs. \eqref{eq:3point-corr-intgrl} is constant for a fixed time $m_{1}$ and the last two terms provide constant value, whereas the second term diverges logarithmically to $\sim 4m_{1}\ln m_{2}$. For $\alpha>2$, all the terms in Eq. \eqref{eq:3point-corr-intgrl}, yields a constant value. This completes the proof of Lemma \ref{lem:m1fixd_t1t2run}. 
%This lemma shows that when $m_2,m_2\rightarrow\infty$, the PACF of the SM gives the exponent, say $w_{2}$, which is equal to $2-\alpha$.
\end{proof}
%\begin{figure}[h]%
%	\centering
%	\subfloat[\centering figure]{{\includegraphics[width=7cm]{fig10a} }}%
%	\qquad
%	\subfloat[\centering figure]{{\includegraphics[width=7cm]{fig12a} }}%
%	\caption{These figures show the log-log plots of the three point position auto-correlation functions of the Slicer map.In figure(a), the graph shows the convergence of auto correlation function $\phi(m_{3},m_{2},m_{1})$ to it's asymptotic scaling, for $\alpha=0.3$. The function $\phi(m_{3},m_{2},m_{1})$ is taken from Eq.\ref{eq.15} and plotted as solid lines of different colours. Similarly, the asymptotic scaling of this correlation function from Eq.~\ref{eq.17} , plotted as dashed lines having the same colour of it's convergence at the same value of $\alpha$. The $m_{1}$ limits large and value of $m_{1},m_{2}$ is legend in the graph. One time is running and rest of the two times are constant. In figure(b), the graph shows the convergence of auto-correlation function $\phi(m_{2}+h,m_{2},m_{1})$ to it's asymptotic behaviour scaling, for $\alpha=0.2$. The function $\phi(m_{2}+h,m_{2},m_{1})$ is taken from Eq.\ref{eq.15} and plotted as solid lines of different colours. Similarly, the asymptotic scaling of this correlation function from Eq.~\ref{eq.18} , plotted as dashed lines having the same colour of it's convergence at the same value of $\alpha$. The $m_{2}$ limits large and value of $m_{1}$ is legend in the graph. $m_{1}$ is constant time $m_{2}$ is running time and the fixed distance between time $m_{2}$ and $m_{3}$ is $h$, such as $h=m_{3}-m_{2}$.}%
%	\label{fig:fig12}%
%\end{figure}
\subsubsection{Scaling of \(\phi_{\alpha}(m_{1},m_{2},m_{3})\) as \(m_{1},\,m_{2},\,m_{3}\to\infty\) with \(m_{2}=m_{1}+\tau\) and \(m_{3}=m_{1}+\tau+\xi\), for fixed $\tau$ and $\xi$}
%The asymptotic behavior of $3$-point PACF $\phi(m_{3},m_{2},m_{1})$ for $ m_{3}\geq m_{2}\geq m_{1}$, with $m_{i} \rightarrow \infty$ where $i\rightarrow \{1,2,3\}$, with time lags $m_{2}-m_{1}=\tau=$ const. and $m_{3}-m_{2}=\xi=$ const., for $\tau, \xi>0$, such as $\tau \leq \xi$ or $\tau > \xi$. The leading order asymptotic behavior is summarized in the following lemmas \ref{lem:univrsl-scaling} and \ref{lem:running_with_q}.
The asymptotic behavior of the $3$-point PACF, \(\phi(m_{1}, m_{2}, m_{3})\), for \(m_{1} \leq m_{2} \leq m_{3}\) with \(m_{i} \rightarrow \infty\) for each \(i \in \{1, 2, 3\}\), with fixed time lags \(m_{2} - m_{1} = \tau \), and \(m_{3} - m_{2} = \xi \), for \(\tau, \xi > 0\), such that \(\tau \leq \xi\) or \(\tau > \xi\), is summarized in the following lemmas: Lemma \ref{lem:univrsl-scaling} and Lemma \ref{lem:running_with_q}.

%%%%%%%%%%%%%%%%%%%%%%%%%%%%%%%%%%%%%%%%%%%%%%%%%%%%%%%%%%%%%%%%%%%%%%%%%%%%%%%%%%%%%%%%%%%%%%%%%%%%
\begin{figure}[t]%
	\centering
	\subfloat[\centering $\alpha=0.5$]%{{\includegraphics[width=8.5cm]{fig13aa-modified}}\label{fig:asympt-match-SM-1}}%
{\hspace{-0.8cm}{\includegraphics[width=9.5cm]{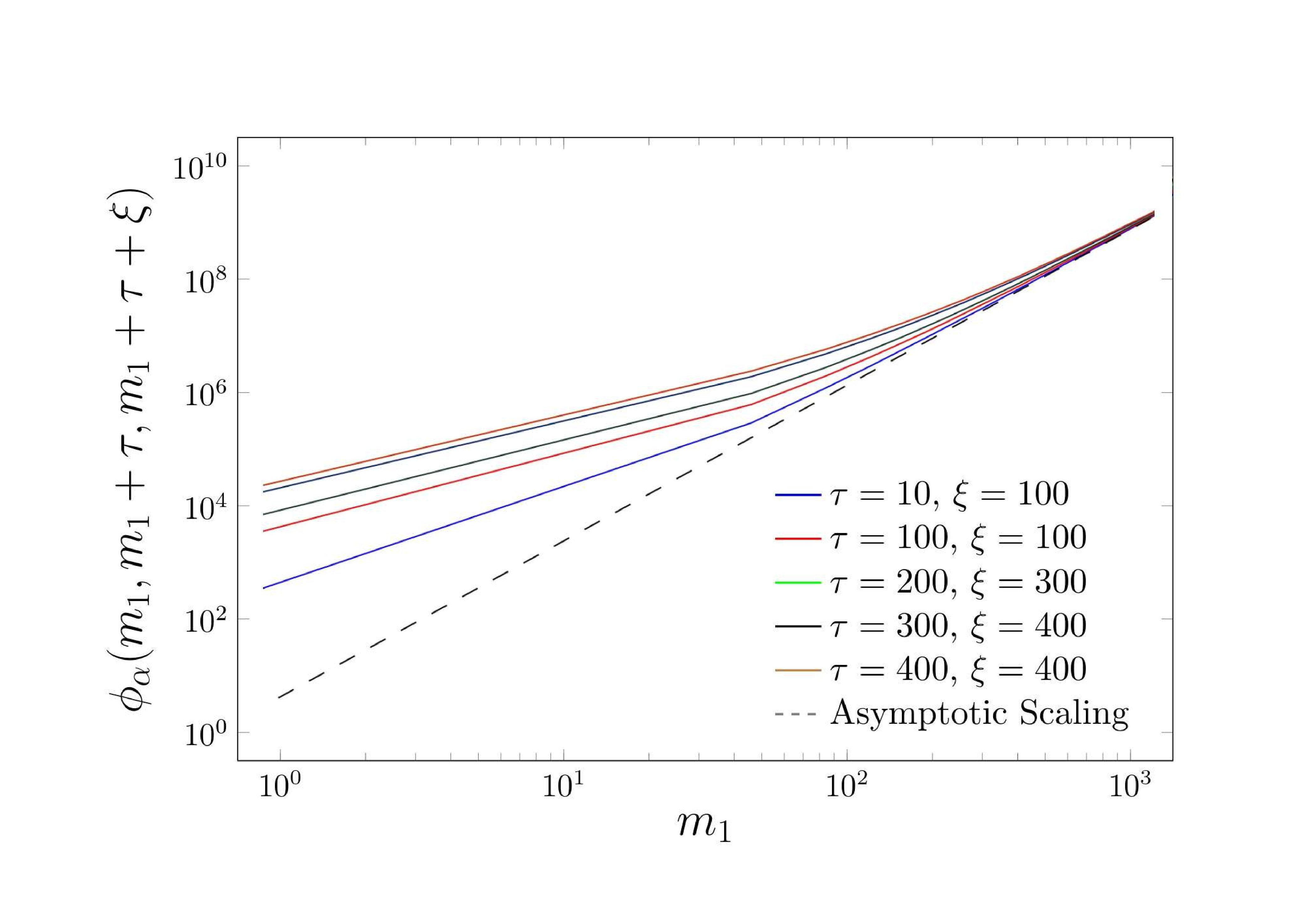}}\label{fig:asympt-match-SM-1}}
	\subfloat[\centering $\alpha=0.2$]{\hspace{-0.8cm}{\includegraphics[width=8.3cm]{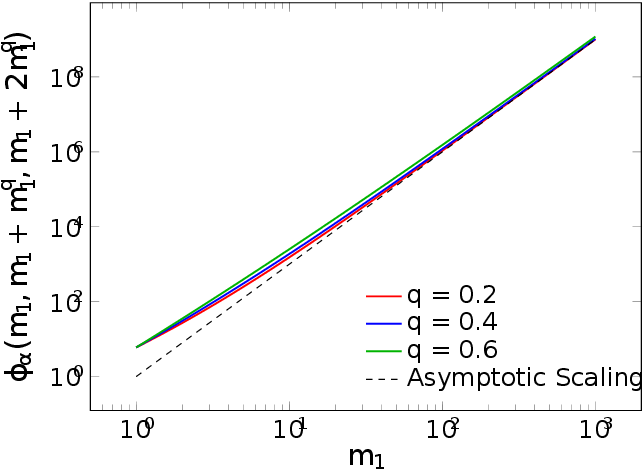}}\label{fig:asympt-match-SM-2}}%
%	fig14a-eps-modified
	\caption{These figures show the log-log plots of the $3$-point PACF of the SM, Eq. \eqref{eq:3point-corr-sum}, along with their asymptotic scaling. Figure \ref{fig:asympt-match-SM-1} shows the convergence of PACF $\phi_\alpha(m_1,m_{1}+\tau,m_{1}+\tau+\xi)$ where $\tau$ and $\xi$ are held constant and $m_1\rightarrow\infty$. The PACF is plotted using solid lines in various colors, representing different combinations of constant differences (see legend). Its respective asymptotic scaling, Eq. \eqref{eq.19}, is plotted with the dotted line for $\alpha=0.5$. Similarly, Fig. \ref{fig:asympt-match-SM-2} represents the convergence plot of PACF $\phi_\alpha(m_{1},m_{1}+\tau,m_{1}+\tau+\xi)$ where $\tau=\xi\sim m_1^q, q<1$, and $m_1\rightarrow\infty$. This is plotted by the solid lines in several colors (see legend) to compare with their asymptotic scaling, Eq. \eqref{eq.19}, for $\alpha=0.2$.\label{fig:fig13}}
\end{figure}
%%%%%%%%%%%%%%%%%%%%%%%%%%%%%%%%%%%%%%%%%%%%%%%%%%%%%%%%%%%%%%%%%%%%%%%%%%%%%%%%%%%%%%%%%%%%%%%%%%%%

\noindent \textbf{Case I}: 
\begin{lemma}\label{lem:univrsl-scaling}
For $m_{i} \rightarrow \infty$ for each $i\in \{1,2,3\}$, the fixed time lags $m_{2}-m_{1}=\tau$ and $m_{3}-m_{2}=\xi$. PACF
$\phi_\alpha(m_{1},m_{1}+\tau,m_{1}+\tau+\xi)$ asymptotically behaves as
\begin{equation}\label{eq.19}
\phi_\alpha(m_{1},m_{1}+\tau,m_{1}+\tau+\xi)\sim\left\{
\begin{array}{lll}
\frac{6}{3-\alpha}m_{1}^{3-\alpha}\,, &\hbox{$0<\alpha< 3$}\,, \\ [2mm]
6\ln(m_{1})\,,&\hbox{$\alpha=3$}\,,\\ [2mm]
\text{const}\,,&\hbox{$\alpha>3$},	
\end{array}
\right.  \; with\;\;
%\end{eqnarray}
%where
%\begin{eqnarray}
\tau , \xi = \left\{
\begin{array}{lll}
const. \,, & \\ [2mm]
m_1^q\,,& q<1 \,,\\ [2mm]
\ln m_1 \,,&\\ [2mm]
\sin m_1\,.	
\end{array}
\right.
\end{equation}
where $\tau$ and $\xi$ are equal or different respectively.
\end{lemma}
\begin{proof}
For \( m_{i} \rightarrow \infty \), where \( i \in \{1,2,3\} \), all terms in Eq. \eqref{eq.16} significantly contributes to the power-law behavior for \( 0 < \alpha < 3 \). Hence, Eq.~\eqref{eq.16} results in the following
\begin{align*}
\phi_\alpha(m_{1},m_{1}+\tau,m_{1}+\tau+\xi)&\sim\frac{6}{3-\alpha}m_{1}^{3-\alpha}\,, \quad \;\;\; \text{for}\;\;\; 0<\alpha<3.
\end{align*}
For $\alpha=3$, the first term in Eqs. \eqref{eq:3point-corr-intgrl} diverges logarithmically as 
%$\phi(m_{1}+h_{1}+h_{2},m_{1}+h_{1},m_{1})
$\sim 6\ln m_{1}$. For $\alpha>3$, all the terms in Eq. \eqref{eq:3point-corr-intgrl}, yields a constant value. 
Hence, these findings complete the proof of Lemma \ref{lem:univrsl-scaling}.
%These are the analytical proof of the lemma.4, this lemma shows that when three times are running with splitting time $h,j$ the auto-correlation of the SM gives the exponent say $w_{3}$, which is equal to $3-\alpha$.
\end{proof}
\noindent Fig. \ref{fig:fig13} illustrates the asymptotic matching of $3$-point PACF $\phi_\alpha(m_1,m_1+\tau,m_1+\tau+\xi)$, Eq. \eqref{eq:3point-corr-sum}, where $\tau$ and $\xi$ are held constant and $\tau=\xi\sim m_1^q,\,q<1$, with their respective asymptotic scaling, Eq. \eqref{eq.19} for $\alpha=0.5$ and $0.2$. The solid colored lines provide a faithful description of asymptotic scaling in the last few decades; hence, numerical evaluation of the definition confirms the power law behavior.

\noindent \textbf{Case II}:
\begin{lemma}\label{lem:running_with_q}
For \(\alpha > 0\) and \( m_i \to \infty \) for each \( i \in \{1, 2, 3\} \), with power-law-type time lag \( \tau = m_2 - m_1 \) and \(\xi =  m_3 - m_2  \), where \( \tau = \xi \sim m_1^q \) and \( q > 1 \), the PACF \( \phi_\alpha \) asymptotically behaves as
\begin{eqnarray}\label{eq.21}
\phi_\alpha(m_{1},m_{1}+m_{1}^{q},m_{1}+m_{1}^{q}+m_{1}^{q})\sim\left\{
\begin{array}{llll}
\frac{4}{2-\alpha}m_{1}^{1+q(2-\alpha)}\,, &\hbox{$0<\alpha< 2$}\,, \\[2mm]
4(q-1)m_{1}\ln(m_{1})\,,&\hbox{$\alpha=2$}\,,\\[2mm]
\frac{2\alpha}{(3-\alpha)(\alpha-2)}m_{1}^{3-\alpha}\,,&\hbox{$2<\alpha<3$},\\	[2mm]
6\ln m_{1}\,, &\hbox{$\alpha=3$}\,,\\[2mm]
\text{const}\,, &\hbox{$\alpha>3$}.
\end{array}
\right.
\end{eqnarray}
\end{lemma}
\begin{proof}
Given \( m_{i} \to \infty \) for each \( i \in \{1, 2, 3\} \), all terms in Eq.~\eqref{eq.16}  significantly contributes to the power-law behavior. When \( 0 < \alpha < 2 \), the second and third terms participate in Eq.~\eqref{eq.16}. By rearranging the equation with a higher order, we obtain
\begin{equation*}
\phi_\alpha(m_{1},m_{1}+m_{1}^{q},m_{1}+m_{1}^{q}+m_{1}^{q})\sim\frac{4}{2-\alpha}m_{1}^{1+q(2-\alpha)},\qquad \text{for}\quad 0<\alpha<2.
\end{equation*}
An analogous derivation for $2<\alpha<3$ which yields
\begin{eqnarray*}
\phi_\alpha(m_{1},m_{1}+m_{1}^{q},m_{1}+m_{1}^{q}+m_{1}^{q}) \sim\frac{2\alpha}{(3-\alpha)(\alpha-2)}m_{1}^{3-\alpha},\qquad \text{for}\quad 2<\alpha<3.
\end{eqnarray*} 
For $\alpha=2$, the second term in Eqs. \eqref{eq:3point-corr-intgrl} diverges logarithmically to $\sim 4(q-1)m_{1}\ln m_{1}$. For $\alpha=3$, the first term in Eq. \eqref{eq:3point-corr-intgrl} diverges logarithmically as $\sim 6\ln m_{1}$. For $\alpha>3$, all terms in Eq. \eqref{eq:3point-corr-intgrl}, yields a constant value. This completes the proof of Lemma \ref{lem:running_with_q}.
\end{proof}

\subsubsection{Scaling of \(\phi_\alpha(m_{1},m_{2},m_{3})\) in the limit \(m_{1}\to\infty\) with \(m_{2}=m_{1}+\tau\), \(m_{3}=m_{1}+\tau+m_{1}^{q}\) (\(q>1\))}
The asymptotic behavior of the $3$-point PACF $\phi_\alpha(m_{1}, m_{2}, m_{3})$ for $m_{i} \rightarrow \infty$ with $i \in \{1, 2, 3\}$ and for $m_{1} \leq m_{2} \leq m_{3}$, with time lags $m_{2} = m_{1} + \tau$ and $m_{3} = m_{1} + \tau + m_{1}^{q}$, where $q > 1$, is summarized in Lemma \ref{lem:scaling_q_large}.
\begin{lemma}\label{lem:scaling_q_large}
For $m_{i} \rightarrow \infty$ with $i \in \{1, 2, 3\}$ and time lags $m_{2} = m_{1} + \tau$ and $m_{3} = m_{1} + \tau + m_{1}^{q}$, the PACF $\phi_\alpha(m_{1}, m_{1} + \tau, m_{1} + \tau + m_{1}^{q})$ asymptotically behaves as
\begin{eqnarray}\label{eq.26}
\phi_\alpha(m_{1},m_{1}+\tau, m_{1} + \tau + m_{1}^{q})\sim\left\{
\begin{array}{llll}
\frac{2}{1-\alpha}m_{1}^{2+q(1-\alpha)}\,, &\hbox{$0<\alpha< 1$}\,, \\[2mm]
2\alpha(q-1)m_{1}^{2}\ln(m_{1})\,,&\hbox{$\alpha=1$}\,,\\[2mm]
\frac{4\alpha}{(3-\alpha)(\alpha-1)}m_{1}^{3-\alpha}\,,&\hbox{$1<\alpha<3$},\\	[2mm]
2\alpha m_{1}\log m_{1}, &\hbox{$\alpha=3$}\,,
\\[2mm]
\text{const.}\,, &\hbox{$\alpha>3$}.
\end{array}
\right.
\end{eqnarray}
\end{lemma}
\begin{proof}
For $m_{i} \rightarrow \infty$ with $i \in \{1, 2, 3\}$, all terms in Eq. \eqref{eq.16} significantly contributes to the power-law behavior. When $0 < \alpha < 1$, the second and third terms participate in Eq.~\eqref{eq.16}. By rearranging the equation, we obtain
\begin{equation*}
\phi_\alpha(m_{1}, m_{1}+\tau, m_{1}+\tau+m_{1}^{q})\sim\frac{2}{1-\alpha}m_{1}^{2+q(1-\alpha)},\qquad \text{for}\quad 0<\alpha<1.
\end{equation*}
For $1<\alpha<3$ first term involves from Eq.~\eqref{eq.16}
\begin{eqnarray}
\phi_\alpha(m_{1}, m_{1}+\tau, m_{1}+\tau+m_{1}^{q})\sim\frac{4\alpha}{(3-\alpha)(\alpha-1)}m_{1}^{3-\alpha},\qquad \text{for}\quad 1<\alpha<3.\nonumber
\end{eqnarray} 
For $\alpha = 1$ and $\alpha = 3$, the first and second terms in Eq. \eqref{eq.16} diverge logarithmically as $\sim 2\alpha(q-1)m_{1}^{2}\ln m_{1}$ and $\sim 6\ln m_{1}$, respectively. For $\alpha > 3$, all the terms in Eq.~\eqref{eq.16} takes constant value. This completes the proof of Lemma \ref{lem:scaling_q_large}.
\end{proof}
%%%%%%%%%%%%%%%%%%%%%%%%%%%%%%%%%%%%%%%%%%%%%%%%
%%  
%\subsubsection{Scaling of $\phi_\alpha(m_{1},m_{2},m_{3})$ with $m_{1}$ fixed for $m_{i}\rightarrow\infty$ for $i\in\{2,3\}$ and $m_{3} = m_{2}+\xi$, where $\xi=m_{2}^{q},$ for $q>1$}
\subsubsection{Scaling of \(\phi_{\alpha}(m_{1},m_{2},m_{3})\) for fixed \(m_{1}\) as \(m_{2},m_{3}\to\infty\) with \(m_{3}=m_{2}+m_{2}^{q}\) (\(q>1\))}
Asymptotic behavior of the $3$-point PACF $\phi_\alpha(m_{1},m_{2},m_{3})$ for $ m_{1}\leq m_{2}\leq m_{3}$, with one fixed time $m_{1}$ and
$m_{i}\rightarrow\infty$ for each $i\in\{2,3\}$. The time lag $m_{3}=m_{2}+m_{2}^{q}$, where $q>1$. The leading-order asymptotic behavior is summarized in Lemma \ref{lem:running_with_q_2tme}.
\begin{lemma}\label{lem:running_with_q_2tme}
For $m_{i} \rightarrow \infty$ with $i \in \{2, 3\}$ and $m_{1}$ fixed, and with power-law type time lag $m_{3} - m_{2} = m_{2}^{q}$, the PACF $\phi_\alpha(m_{1}, m_{2}, m_{2} + m_{2}^{q})$ asymptotically behaves as
\begin{eqnarray}\label{eq.27}
\phi_\alpha(m_{1},m_{2},m_{2}+m_{2}^{q})\sim\left\{
\begin{array}{llll}
\frac{2m_{1}}{1-\alpha}m_{2}^{1+q(1-\alpha)}\,, &\hbox{$0<\alpha< 1$}\,, \\[2mm]
2m_{1}m_{2}(q-1)\ln m_{2}\,,&\hbox{$\alpha=1$}\,,\\[2mm]
\frac{2\alpha m_{1}}{(\alpha-2)(1-\alpha)} m_{2}^{2-\alpha}\,,&\hbox{$1<\alpha<2$}\,,\\	[2mm]
4m_{1}\ln m_{2}\,, &\hbox{$\alpha=2$}\,,\\[2mm]
\text{const.}, &\hbox{$\alpha>2$}.
\end{array}
\right.
\end{eqnarray}
\end{lemma}
\begin{proof}
For a fixed \( m_1 \), the first term in Eq. \eqref{eq.16} does not significantly contribute to the power-law behavior, whereas the second and third terms contribute to asymptotic behavior. From Eq.~\eqref{eq.16}, for \( 0 < \alpha < 2 \), we have
\begin{align*}
\phi_\alpha(m_{1},m_{2},m_{2}+m_{2}^{q}) \sim \frac{2m_{1}}{1-\alpha}m_{2}^{1+q(1-\alpha)}\,, 
\qquad \text{for}
\quad 0<\alpha<1.  
\end{align*}
For \(\alpha = 1\), Eq. \eqref{eq:3point-corr-intgrl} exhibits logarithmic scaling, given by \(\sim 2m_{1}m_{2}(q-1)\ln m_{2}\). 
%This indicates that \(\phi_\alpha\) grows logarithmically with \(m_2\), with its rate of growth modulated by \(m_1\) and \(q\).
For \(\alpha > 2\), all terms in Eq. \eqref{eq.16} takes on constant values, indicating that the behavior of the system no longer depends on \(m_2\) and the scaling effects vanish,
%, leading to a steady-state or saturation point in the system's behavior. 
which completes the proof of Lemma \ref{lem:running_with_q_2tme} by demonstrating the different scaling behaviors.
\end{proof}
%%%%%%%%%%%%%%%%%%%%%%%%%%%%%%%%%%%%%%%%%%%%%%%%%%%%
%The scaling for $\alpha=1$ Eq.~\eqref{eq.16} amounts to logarithmic scaling as $\phi(m_{2}+m_{2}^{q},m_{2},m_{1}) \sim 2m_{1}m_{2}(q-1)\ln m_{2}$. For $\alpha>2$, all terms from Eq.~\eqref{eq.16} take constant values.
%%%%%%%%%%%%%%%%%%%%%%%%%%%%%%%%%%%%%%%%%%%%%%%%%%
\begin{remark}  
For $\alpha > 0$, the $3$-point PACF $\phi_\alpha(m_1, m_2, m_3)$ recovers the asymptotic scaling of the position moments, as shown in Eq. \eqref{eq.10} for $j = \{1, 2, 3\}$:  
\begin{equation*}  
\phi_\alpha(m_1, m_2, m_3) \sim \langle |\Delta x_m|^j \rangle.  
\end{equation*}  
\begin{description}  
		\item[$j=1:$]  
		For the $3$-point PACF $\phi_\alpha(m_1, m_2, m_3)$, where $0<\alpha<1$, $m_1$ and $m_2$ are fixed, and $m_3 \to \infty$, Eq. \eqref{eq.17}, exhibits the same asymptotic scaling as the first moment of displacement $\langle |\Delta x_m| \rangle$, [\cf~Eq. \eqref{eq.10} for $j = 1$].  
		
		\item[$j=2:$]  
For the $3$-point PACF $\phi_\alpha(m_1, m_2, m_2 + \xi)$, where $m_1$ is fixed, $m_2 \to \infty$  and  
\begin{equation*}
\xi = \left\{
\begin{array}{lll}
\text{const.} \,, & 0<\alpha<2, & \\ [2mm]
m_2^q\,,& 1<\alpha<2, & q>1 \,,
\end{array}
\right.
\end{equation*}
see Eqs. \eqref{eq.18} and \eqref{eq.27} respectively. In both cases, the PACF exhibits the same asymptotic scaling as the second moment of displacement $\langle |\Delta x_m|^2 \rangle$, \ie the MSD [\cf~Eq. \eqref{eq.10} for $j = 2$].

		\item[$j=3:$]  		
					For the $3$-point PACF \(\phi_\alpha(m_1, m_1 + \tau, m_1 + \tau + \xi)\), for \(0 < \alpha < 3\), where \(m_1 \to \infty\), and \(\tau, \xi\) are the time lags defined in Eq. \eqref{eq.19}, and in another case for \(1 < \alpha < 2\) when the time lag \(\tau\) is constant and \(\xi \sim m_1^q, \; q > 1\) (see Eq. \eqref{eq.26}), both exhibit the same asymptotic scaling as the third moment of displacement \(\langle |\Delta x_m|^3 \rangle\) [\cf~Eq. \eqref{eq.10} for \(j = 3\)].
\end{description}  
\end{remark} 
%%%%%%%%%%%%%%%%%%%%%%%%%%%%%%%%%%%%%%%%%%%%%%%%%%
\begin{remark}
Given $\alpha > 0$, the $j^{\text{th}}$ position moments 
$\langle |\Delta x_m|^j \rangle$, as in Eq. \eqref{eq.10} and the generalized PACF $\phi_{\alpha}(m_1, m_2, \dots, m_j)$ for $0 < \alpha < j$, as in Eq. \eqref{eq.30}, follows the relation of time lags $\tau_{l-1} = m_l - m_{l-1}$ for $l = 2, 3, \dots, j$, where $\tau_{l-1}$ is either fixed or asymptotically $\sim m_1^q$ with $q < 1$. In the asymptotic regime $m,m_1\rightarrow\infty$, both exhibit the same scaling as
\begin{align}
\langle |\Delta x_m|^p \rangle \sim   \phi_{\alpha}(m_1, m_2, \dots, m_j) \,, \qquad p = 1, 2, 3, \dots, j\,. 
\end{align}
\end{remark}
%%%%%%%%%%%%%%%%%%%%%%%%%%%%%%%%%%%%%%%%%%%%%%%%%%%%%%%%%%%%%%%%%%%%%%%%%%%%%%%%%%%%%%%%%%%%%%%%%%%%
\section{L\'evy-Lorentz gas and statistical equivalence to the Slicer Map}\label{sec.2}
%As the Slicer map is explained in the above section so, for the comparison with the L\'{e}vy-Lorentz gas (LLg) define the LLg. 
The LLg as is a type of one-dimensional random walk in a quenched environment, where the surroundings are fixed, and particles move ballistically between random static scatterers with velocity ($\pm$v). Probability $1/2$ either transmits or reflects particles in such a setting. The distances $r$ between neighboring scatterers are independently and identically distributed random variables sampled from a L\'evy distribution with a given probability density
\begin{equation*}
\lambda(r) = \beta r_{0}^{\beta} \frac{1}{r^{\beta+1}}\;\;\; r\in[r_{0},+\infty),\qquad \beta>0\,,
\end{equation*}
where $r_{0}$ is the cut-off for fixing the scale length of the system. 
Barkai et al.~\cite{barkai2000one} calculated the bounds of the MSD for the equilibrium and non-equilibrium initial conditions. Subsequently, Burioni et al.~\cite{burioni2010levy} adopted some simplifying assumptions to find the asymptotic form  for non-equilibrium conditions of all moments $\langle |r(t)|^{p}\rangle$ with $p>0$\,,
\begin{eqnarray}\label{eq:MOMBUR}
  \langle |r(t)|^{p} \rangle
  \sim
  \left\{
    \begin{array}{lll}
      t^{\frac{p}{1+\beta}} \,,                 & \quad\text{for  } & \beta<1,\ p<\beta    \, ,\\
      t^{\frac{p(1+\beta)-\beta^{2}}{1+\beta}} \,,     & \quad\text{for  } & \beta<1,\ p>\beta    \, ,\\
      t^{\frac{p}{2}} \,,                    & \quad\text{for  } & \beta>1,\ p<2\beta-1 \, ,\\
      t^{\frac{1}{2}+p-\beta} \,,                & \quad\text{for  } & \beta>1,\ p>2\beta-1 \, . 
    \end{array}
  \right. 
\end{eqnarray}
%More recently, Bianchi et al.~\cite{BCLL16} presented a rigorous mathematical study to prove a central limit theorem, and M. Zamparo~\cite{Z23} investigated large fluctuations and transport properties of the LLg.

For $p=2$, the MSD implies
\begin{eqnarray}\label{eq:LLgMeanSquare}
  \langle r(t)^{2} \rangle 
  \sim 
  t^\xi, \;\;\;\;
  %\quad\text{with}\quad 
  \xi
  =
%  \left\{
%    \begin{array}{llr@{\;\xi\;}l}
%      2 - \frac{\xi^{2}}{(1+\xi)}   & \quad\text{for  } & &  < 1    \, ,\\
%      \frac{5}{2}-\xi              & \quad\text{for  } & 1   \leq &  < 3/2  \, ,\\
%      1                            & \quad\text{for  } & 3/2 \leq &         \, .
%    \end{array}
%  \right. 
\left\{
    \begin{array}{lll}
      2 - \frac{\beta^{2}}{(1+\beta)} \,,   & \quad\text{for  } & \beta < 1    \, ,\\
      \frac{5}{2}-\beta \,,             & \quad\text{for  } & 1  \leq \beta < \frac{3}{2}  \, ,\\
      1  \,,                          & \quad\text{for  } & \frac{3}{2} \leq \beta         \, .
    \end{array}
  \right. 
\end{eqnarray}%
Unlike the SM, which exhibits sub-diffusive transport for \( \alpha > 1 \), non-equilibrium initial conditions for the LLg only result in super-diffusive (\( 0 < \xi < 3/2 \)) or diffusive (\( \xi \geq 3/2 \)) regimes; sub-diffusion is not expected. In the super-diffusive regime (\( 0 < \alpha < 1 \)), the moments of the SM can be mapped to those of the LLg \cite{salari2015simple, giberti2019equivalence}. All the moments of the SM, Eq. \eqref{eq.10}, a scale similar to those conjectured and numerically validated for the LLg, Eq.~\eqref{eq:MOMBUR} once the second moments do. This holds if the following condition is met with \( p = 2 \) (\cf~Eq.~\eqref{eq.10})
\begin{eqnarray}\label{eq.39}
\alpha =\left\{
\begin{array}{lll}
\frac{\beta^{2}}{1+\beta}, & \mbox{for}&\hbox{$0<\beta\leq(p-1)$}, \\ [2mm]
\beta - \frac{1}{2}, & \mbox{for}&\hbox{$(p-1)<\beta\leq \frac{p+1}{2}$},\\ [2mm]
\frac{p}{2}, & \mbox{for}&\hbox{$\beta>\frac{p+1}{2}$}.	
\end{array}
\right.
\end{eqnarray}
When adopting the above relation (\cf~Eq.~\eqref{eq.39}), the two distinct processes share the same position moments and low-order PACF despite their apparent differences, which suggests a fundamental connection or underlying universal behavior between them. This indicates that these processes are indistinguishable in terms of moments and low-order PACF \cite{salari2015simple, giberti2019equivalence, vollmer2021displacement}. Next, we extend the equivalence to the scaling forms of the $3$-point PACF and investigate whether each analytically derived correlation scaling of the SM demonstrates the same strong equivalence or deviates to any degree. The various scaling forms of the $3$-point PACF are calculated analytically for the SM in section \ref{sec:corr-scalings}. Subsequently, we compare each of them with the numerically estimated PACF of the LLg, for which an analytical expression is still an open question.  
\subsection{Numerical estimation of position correlation function of the L\'evy-Lorentz gas}
The generalized (or $n$-point) PACF of the LLg for the position of trajectories at different times can be defined as
%For the $n$-point correlations, for this we have to know the continuous time position of the particles on $n$-times such as $r(t_{1}),r(t_{2}),r(t_{3}),\cdots, r(t_{n})$. The average of $n$-point correlations for  fixed scatterers $y$ is $\mathbb{E}_{w}[r(t_{1})r(t_{2}),\cdots,r(t_{n})|y]$ where $n\in\mathbb{N}$.
%Average this over the realization scatterers which yields the numerical correlations for $\text{LLg}$
\begin{equation}\label{eq:npointcorrllg}
\Phi_\beta(t_{1},t_{2},\cdots , t_{n})= \langle(r(t_{1})r(t_{2})\cdots r(t_{n} \rangle)_{\beta} = \mathbb{E}[(r(t_{1})r(t_{2})\cdots r(t_{n})],
\end{equation} 
where $\beta$ denotes the scatterers distribution or diffusion parameter, $\langle \cdot \rangle$ and $\mathbb{E}$ represent the averages, first averaging over the particles and then on the given random scatterer realization. In the following section, we compare the analytically computed scaling of the $3$-point PACF of the SM with the correlations of the LLg. Therefore, we define the $3$-point PACF of the LLg in section \ref{subsec.7}, and compare the numerically estimated power-law behavior with that of the SM correlation scaling.
\begin{lemma}\label{lem:npoint_pwr_LLg}
For \( \beta > 0 \), and \( t_i \to \infty \) for each \( i \in \{1, 2, \dots, j\} \), the generalized PACF \( \Phi_\beta \), Eq.~\eqref{eq:npointcorrllg}, of the LLg has the following asymptotic form
\begin{equation}\label{eq:np-corr-momnt} 
\Phi_\beta(t_1,t_2,\cdots,t_j) \sim K(\tau,\xi,\dots,\zeta) \;t_\mathcal{I}^{\nu_p}\,, 
\end{equation} 
for \( p = 2, 3, \dots, j \), where \( (\tau, \xi, \dots, \zeta) \) represents the time lags, \( K(\tau, \xi, \cdots, \zeta) \) denotes the pre-factor, and \( t_\mathcal{I} \) corresponds to any of $t_1,t_2,\cdots,t_j$, depending on the selected time lags. \( \nu_p \) is the power law exponent for the respective order of the correlation function, which is obtained by fitting the data.
\end{lemma}
%%%%%%%%%%%%%%%%%%%%%%%%%%%%%%%%%%%%%%%%%%%%%%%%%%
\subsection{$3$-point position auto correlation function of the L\'{e}vy Lorentz gas}\label{subsec.7}
We define the $3$-point PACF for $t_{i}>0$, for each $i\in \{1,2,3\}$, by requesting $n=3$ in Eq. \eqref{eq:npointcorrllg}, and then by Lemma \ref{lem:npoint_pwr_LLg} the correlation is given as 
\begin{equation} \Phi_\beta(t_{1},t_{2},t_{3})={\langle r(t_{1})r(t_{2})r(t_{3})\rangle}_{\beta} \sim K(\tau,\xi) \;t_\mathcal{I}^{\nu_3}\,,
\end{equation} 
%\quad\text{with}\quad t_{3}\geq t_{2}\geq t_{1} 
$\langle \cdot \rangle$ represents the average distance over times. 
%This can be find by numerical simulation by considering average of scatterers particle given $y$ over the average of realization scatters. it yields
%\begin{equation*}
%\Phi_\beta(t_{3},t_{2},t_{1})=\mathbb{E}_{y}[\mathbb{E}_{w}(r(t_{1})r(t_{2})r(t_{3}))]\quad\text{with}\quad t_{3}\geq t_{2}\geq t_{1}
%\end{equation*}
%This gives the $3$-point PACF of the LLg. 
The primary objective is to compare the numerically estimated LLg correlations with the scaling behavior of the SM, for which the analytical results are presented in section \ref{sec:corr-scalings}. In that section, we calculated the asymptotic scaling of the $3$-point PACF of the SM for various time relationships. Here, we employ the same time composition as that used previously (\cf~Sec. \ref{sec:corr-scalings}):
\begin{description}
\item[1.] 
	\(\Phi_\beta(t_{1}, t_{2}, t_{3})\) for \(t_{3}\rightarrow\infty\) with \(t_{1}\) and \(t_{2}\) fixed.
\item[2.]
	\(\Phi_\beta(t_{1}, t_{2}, t_{2} + \xi)\) for \(t_{2}\rightarrow\infty\) with \(\xi > 0\) and \(t_{1}\) fixed.
\item[3.]
	\(\Phi_\beta(t_{1}, t_{1} + \tau, t_{1} + \tau + \xi)\) for \(t_{1}\rightarrow\infty\) with \(\tau, \xi > 0\) as constants.
\item[4.]
	\(\Phi_\beta(t_{1}, t_{1} + t_{1}^{q}, t_{1} + 2t_{1}^{q})\) for \(t_{1}\rightarrow\infty\) with \(0 < q < 1\).
\item[5.]
	\(\Phi_\beta(t_{1}, t_{1} + \ln t_{1}, t_{1} + 2\ln t_{1})\) for \(t_{1}\rightarrow\infty\).
\item[6.]
	\(\Phi_\beta(t_{1}, t_{1} + \sin t_{1}, t_{1} + 2\sin t_{1})\) for \(t_{1}\rightarrow\infty\).
\item[7.]
	\(\Phi_\beta(t_{1}, t_{1} + \tau, t_{1} + \tau + t_{1}^{q})\) for \(t_{1}\rightarrow\infty\) with \(\tau\) and \(q\) as constants and \(0 < q < 1\).
\end{description} 
For all these cases, we will compare the asymptotic scaling of the SM with \(\text{LLg}\).
%%%%%%%%%%%%%%%%%%%%%%%%%%%%%%%%%%%%%%%%%%%%%%%%%%%%%%%%%%%%%%%%%%%%%%%%%%%%%%%%%%%%%%%%%%%%%%%%
\begin{figure}[t]%
	\centering
	\subfloat[\centering Varying $\beta$]{{\includegraphics[width=0.51\linewidth]{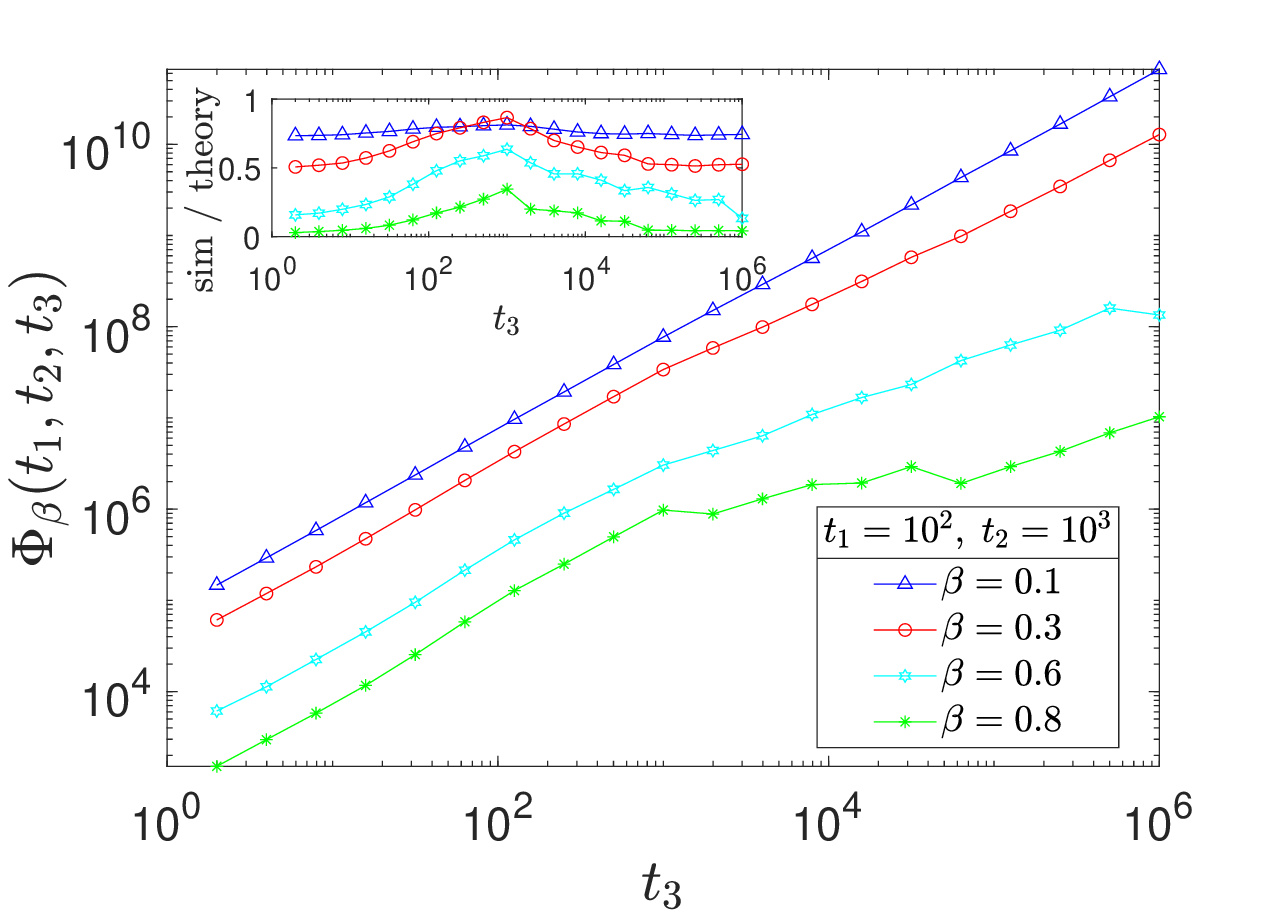} } \label{fig:m1m2fixed_a} }%
	\subfloat[\centering $\beta=0.1$]{{\includegraphics[width=0.51\linewidth]{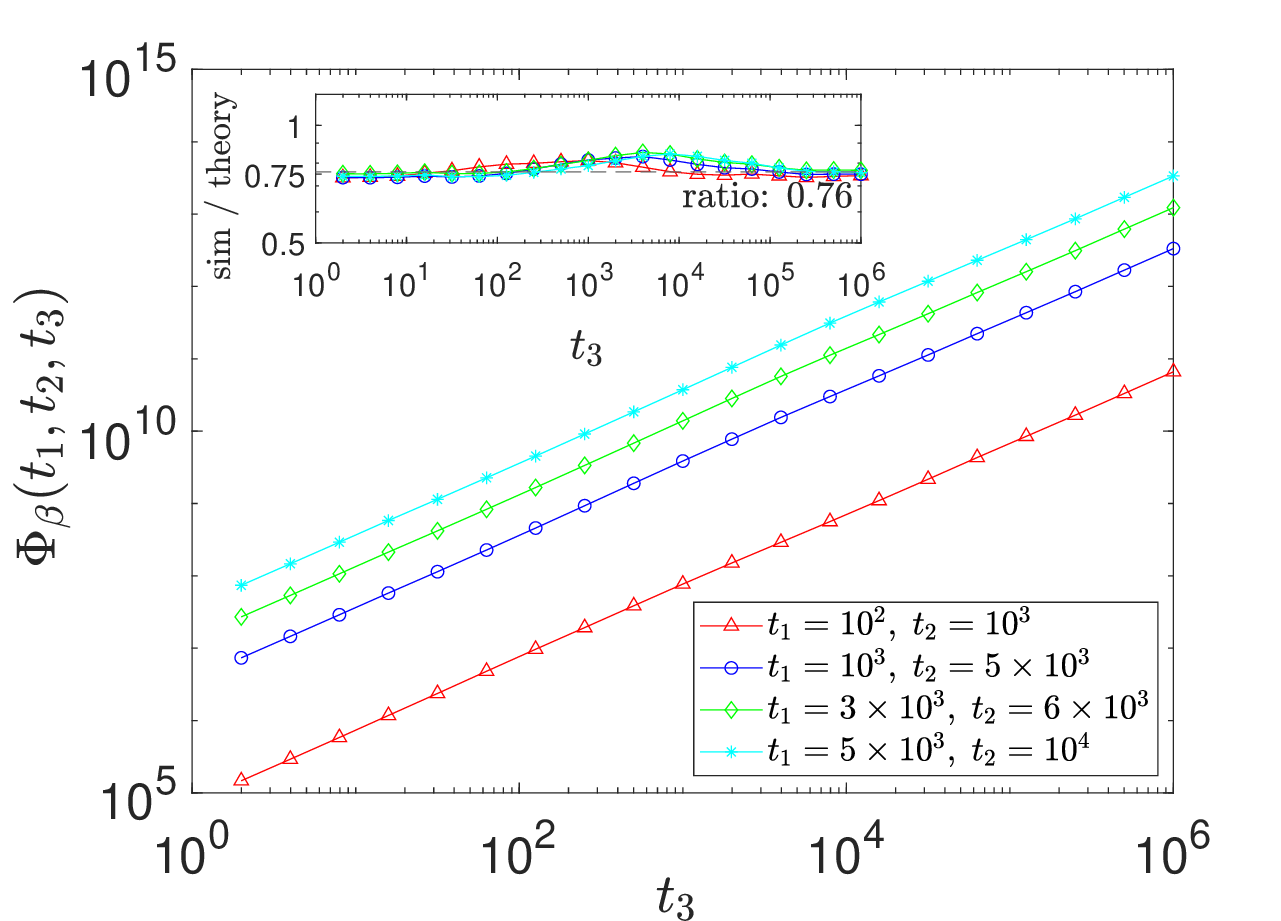} } \label{fig:m1m2fixed_b}}%
	\caption{\label{fig:m1m2fixed}These figures present the log-log plots of the $3$-point PACF, $\Phi_\beta(t_{1},t_{2},t_{3})$ as a function of $t_3$, with $t_1,t_2$ fixed and $t_3\rightarrow\infty$. In Figure \ref{fig:m1m2fixed_a}, the $3$-point LLg correlation data is plotted for several values of $\beta$, as indicated in the legend, and the slopes of the best-fit lines demonstrate the power-law behavior of the LLg.
Similarly, Fig. \ref{fig:m1m2fixed_b} focuses on $\beta=0.1$, presenting data for different fixed combinations of times $t_1$ and $t_2$ along with their corresponding best fit to the data. In Fig. \ref{fig:m1m2fixed}, the two insets show the ratio of the numerically estimated LLg correlation $\Phi_\beta$ (\cf~Eq.~\eqref{eq:m1m2fix_plaw_llg}) to the correlation scaling value $\phi_\alpha$ (\cf~Eq.~\eqref{eq:m1m2fix_plaw_SM}) of the SM which was calculated analytically.}%	
\end{figure}
%%%%%%%%%%%%%%%%%%%%%%%%%%%%%%%%%%%%%%%%%%%%%%%%%%%%%%%%%%%%%%%%%%%%%%%%%%%%%%%%%%%%%%%%%%%%%%%%
\subsubsection{Scaling test of the correlation $\Phi_\beta(t_{1},t_{2},t_{3})$ with $t_{1},t_{2}$ are fixed}
%%%%%%%%%%%%%%%%%%%%%%%%%%%%%%%%%%%%%%%%%%%%%%%%%%
In this section, we compare the scaling behavior of the $3$-point PACF for the LLg with the SM correlation under the same time configuration. For the SM, the asymptotic scaling form \(\phi_\alpha(m_1, m_2, m_3)\) is given by Eq.~\eqref{eq.17}. By fixing \(m_1, m_2\), and \(0 < \alpha < 1\), we obtain
\begin{equation}\label{eq:m1m2fix_plaw_SM}
\phi_\alpha(m_1, m_2, m_3) \sim \frac{m_1 m_2}{1-\alpha} m_3^{1-\alpha},\quad \text{as}\quad m_3 \rightarrow \infty.
\end{equation}
For the LLg, numerical simulations yield data for \(\Phi_\beta(t_1, t_2, t_3)\), plotted against \(t_3\) in Fig. \ref{fig:m1m2fixed}. Figure \ref{fig:m1m2fixed_a} shows the results for fixed \(t_1=10^2,\, t_2=10^3\) for several values of \(\beta\) between \(0.1\) and \(0.8\). In contrast, Fig.~\ref{fig:m1m2fixed_b} displays data for fixed \(\beta = 0.1\) across four combinations of \(t_1\) and \(t_2\). In both cases, the slopes of the fitted lines in the final decades confirm a clear power-law behavior, supporting the asymptotic scaling relation
\begin{eqnarray}\label{eq:m1m2fix_plaw_llg}
\Phi_\beta(t_1, t_2, t_3) \sim K\,t_3^{\nu_3}, \quad t_3 \rightarrow \infty.
\end{eqnarray}
The comparison shows that the scaling exponents  $1-\alpha$ and $\nu_3$ from Eqs.~\eqref{eq:m1m2fix_plaw_SM} and \eqref{eq:m1m2fix_plaw_llg}, respectively, are closely aligned for small values of \(\beta\), as shown in Fig.~\ref{fig:m1m2fixed_b}. For a larger \(\beta\), the agreement remains fair, considering the numerical accuracy of the data (Fig.~\ref{fig:m1m2fixed_a}). Notably, the insets in Fig.~\ref{fig:m1m2fixed} illustrates the ratio between the simulated and theoretical results, revealing minor discrepancies within the interval \((0,1)\), which further validates the accuracy of our numerical approach.
This agreement underscores the equivalence of the scaling laws, despite originating from distinct microscopic dynamics, provided that the \((\alpha, \beta)\) functional relationship is properly tuned (\cf~Eq.~\eqref{eq.39} for \(p=2\)). This highlights the robustness of the theoretical framework in capturing scaling behavior across different systems.

\begin{figure}[t]%
	\centering
	\subfloat[\centering varying $\beta$]{{\includegraphics[width=8.2cm]{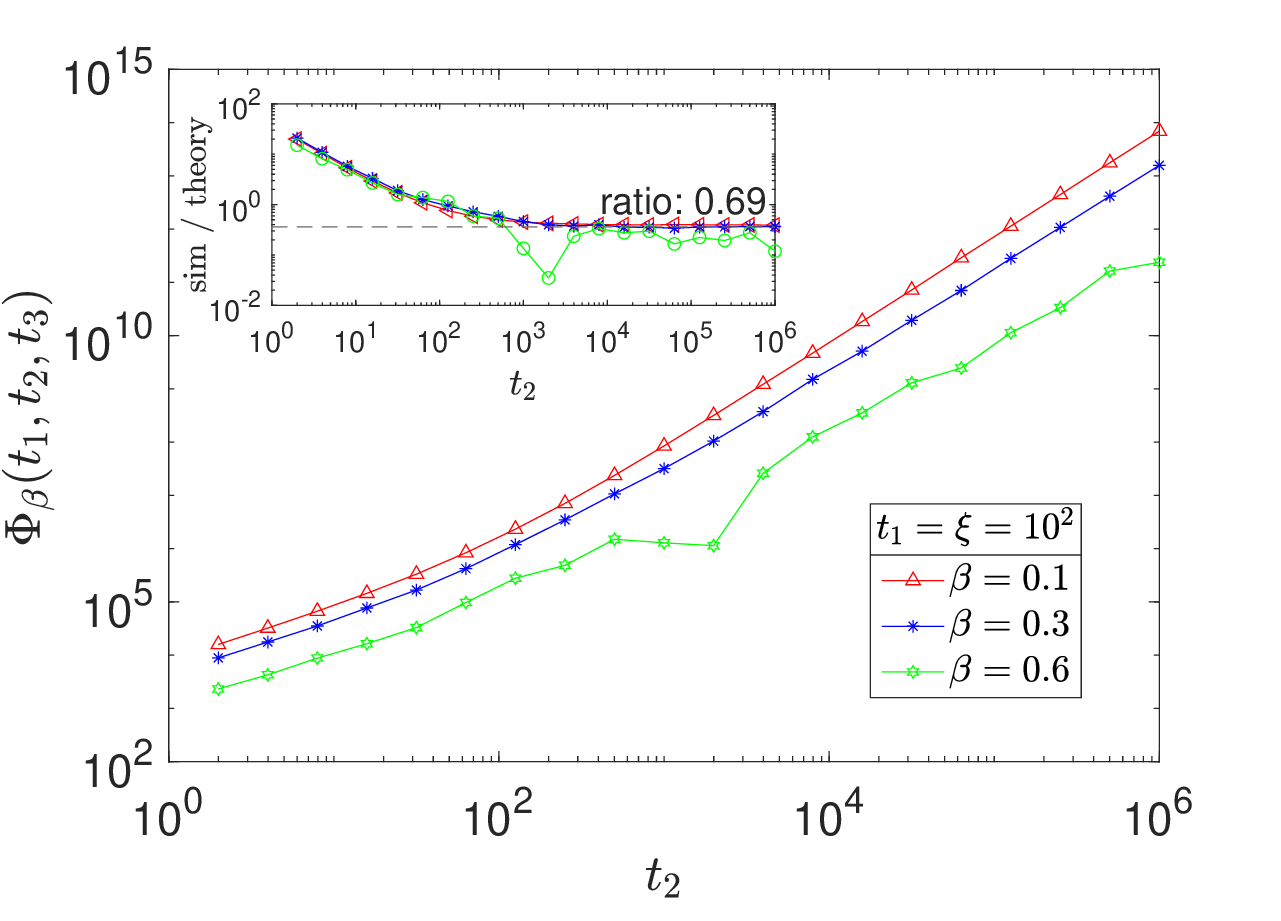} }\label{fig:t1fixd_a}}%
	\subfloat[\centering $\beta=0.1$]{{\includegraphics[width=8.2cm]{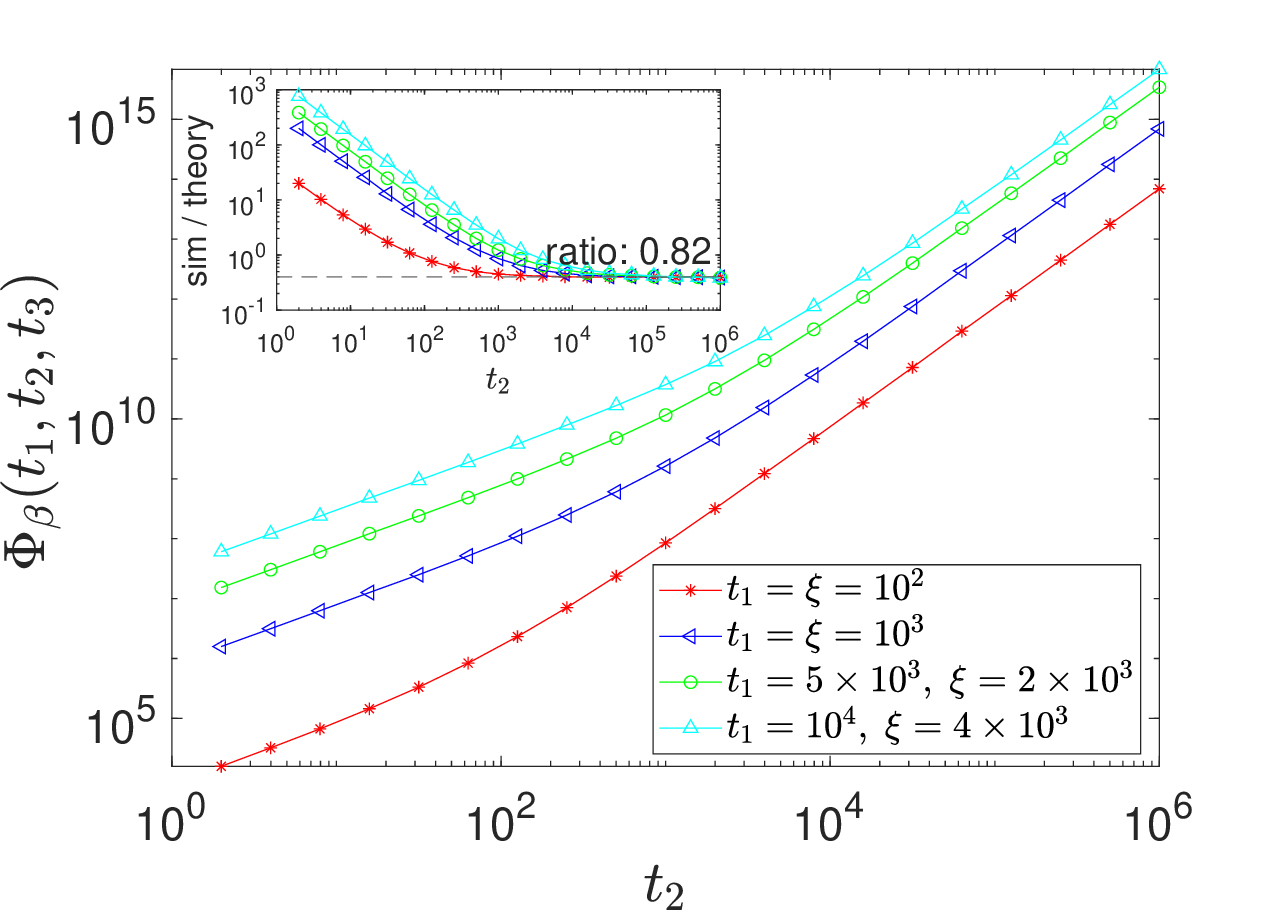} }\label{fig:t1fixd_b}}%
	\caption{\label{fig:t1fixd} 
These figures show the log-log plots of the $3$-point PACF $\Phi_\beta(t_1, t_2, t_2+\xi)$ as a function of $t_2$, where $t_1$ and $\xi$ are fixed, and $t_2 \to \infty$. Figure \ref{fig:t1fixd_a} represents the $3$-point PACF $\Phi_\beta$ of the LLg for several values of $\beta$ with fixed values $t_1 = \xi = 10^2$, as indicated in the legend. The slopes of the best-fit lines demonstrate the power-law behavior of the LLg. Similarly, in Figure \ref{fig:t1fixd_b}, the data for $\beta = 0.1$ with several fixed combinations of $t_1$ and $\xi$ are shown, along with their respective best-fit lines. The two insets display the ratio of the numerically estimated LLg correlation $\Phi_\beta$ (\cf~Eq.~\eqref{eq:m1fxd_xi_LLg}) to the correlation scaling $\phi_\alpha$ (\cf~Eq.~\eqref{eq:m1fxd_xi_SM}) of the SM, which was computed analytically. 
%Figure(a)shows the log-log plot with fixed value of constant time as mentioned in the graph and different values of the splitting time constant say $\tau$, whose values are also mentioned in the legends. The LLg parameter($\beta$) varies and mentioned in this figure as legends. The slopes of data fit line through the data points mentioned in the graph shows the power law behaviors for the LLg as exponents $(E_2(0.1),E_2(0.3),E_2(0.6))$. While on the other hand, in figure(b) the parameter has fixed value and the constant varies as mentioned in the graph along with the slope data as exponent $(E_2(0.1a),E_2(0.1b),E_2(0.1c),E_2(0.1d))$. This is named as the asymptotic scaling of $3$-point PACF for the LLg.
}
\end{figure}
%%%%%%%%%%%%%%%%%%%%%%%%%%%%%%%%%%%%%%%%%%%%%%%%%%%%%%%%%%%%%%%%%%%%%%%%%%%%%%%%%%%%%%%%%%%%%%%%
\subsubsection{Scaling test of $\Phi_\beta(t_{1},t_{2},t_{2}+\xi)$ $t_1$ fixed, and $\xi>0$}

In this section, we compare the numerically estimated $3$-point PACF of the LLg, 
\(\Phi_\beta(t_{1}, t_{2}, t_{2}+\xi)\), where \(t_{1}\) and \(\xi\) take several fixed values and \(t_{2} \rightarrow \infty\), with the scaling of the SM correlation derived in section \ref{sec:m1fixd_m2_run_SM}, given by Eq. \eqref{eq.18}. For \(0 < \alpha < 2\), we have  
\begin{equation}\label{eq:m1fxd_xi_SM} 
\phi_\alpha(m_{1},m_{2},m_{2}+\xi)\sim \frac{4m_{1}}{2-\alpha} m_{2}^{2-\alpha},\quad \quad m_{2}\rightarrow\infty \,.
%\quad\text{and}\quad h>0 \quad\text{constant}.
\end{equation}
We collect the LLg correlation data for multiple values of \(\beta\), with some fixed \(t_1 \) and constant time lags \(\xi\). In Fig. \ref{fig:t1fixd}, \(\Phi_\beta(t_1, t_2, t_2+\xi)\) are plotted as a function of \(t_2\), and the results for varying \(\beta\) ranging from $0.1$ to $0.6$, with \(t_1 = 10^2\) and \(\xi = 10^2\), are shown in Fig.~\ref{fig:t1fixd_a}. Similarly, for a \(\beta = 0.1\), the data is presented for four different fixed values of \(t_1\) and time lags \(\xi\) in Fig.~\ref{fig:t1fixd_b}. Consequently, the PACF of the LLg asymptotically follows  
\begin{eqnarray}\label{eq:m1fxd_xi_LLg}
\Phi_{\beta} (t_1,t_2,t_2+\xi) \sim K\, t_{2}^{\nu_3},\qquad t_2 \rightarrow \infty.
\end{eqnarray} 
The SM correlation scaling matches the numerically estimated LLg correlation exponent \(\nu_3\) when the \((\alpha,\beta)\) is properly tuned (\cf~Eq. \eqref{eq.39} for \(p=2\)).  
For smaller values of \(\beta\), the numerical results show strong agreement with the theoretical predictions (see Fig.~\ref{fig:t1fixd}), while for larger \(\beta\), the agreement is satisfactory (see Fig.~\ref{fig:t1fixd_a}). The insets of Fig.~\ref{fig:t1fixd} shows the ratio of the numerical simulations to the theoretical results for the LLg and the SM (\cf~Eqs.~\eqref{eq:m1fxd_xi_LLg} and \eqref{eq:m1fxd_xi_SM}). A small deviation, confined to the interval \((0,1)\), confirms the accuracy of our numerical results.  
%%%%
\begin{figure}[t] %ht
\hspace{-8mm}
\begin{subfigure}{.5\textwidth}
  \centering
  % include first image
  \includegraphics[width=1.12\linewidth]{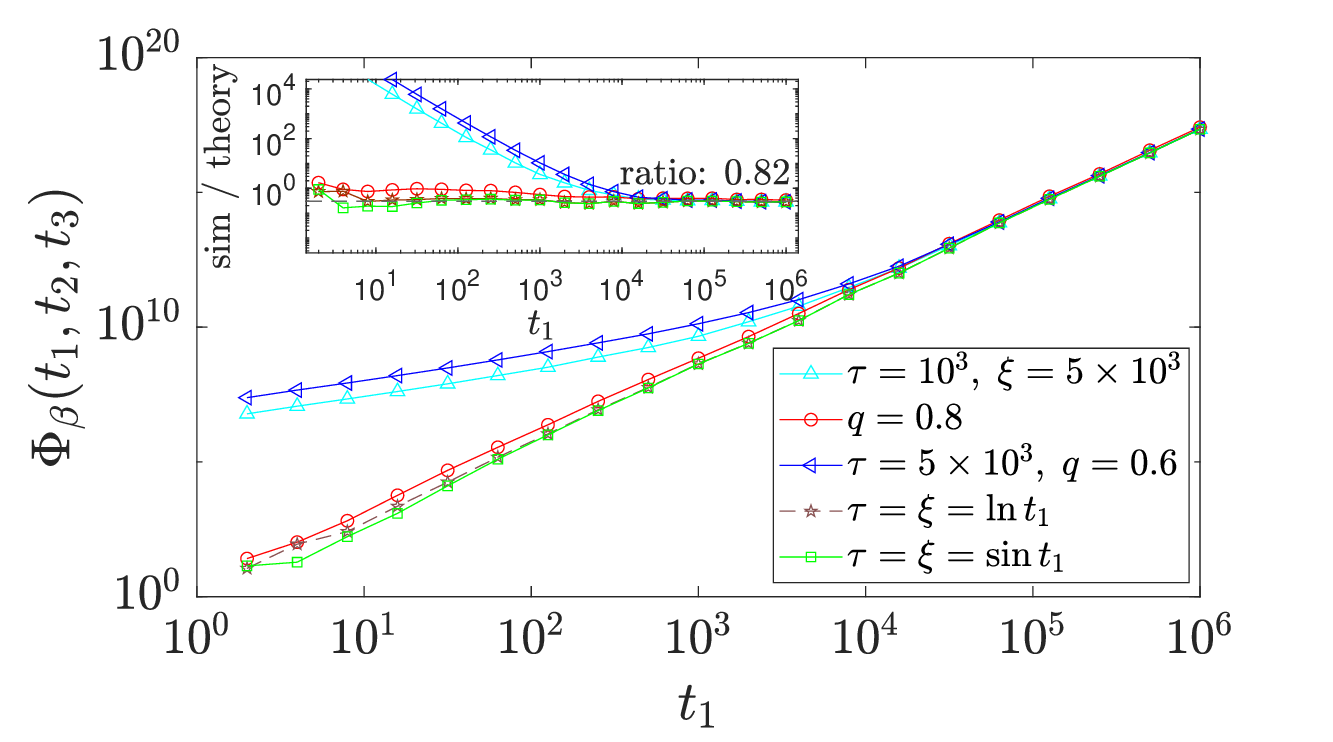}  
  \caption{different time lags for $\beta=0.3$}
  \label{fig:sub-first}
\end{subfigure}
\;
\begin{subfigure}{.5\textwidth}
  \centering
  % include second image
  \includegraphics[width=1.11\linewidth]{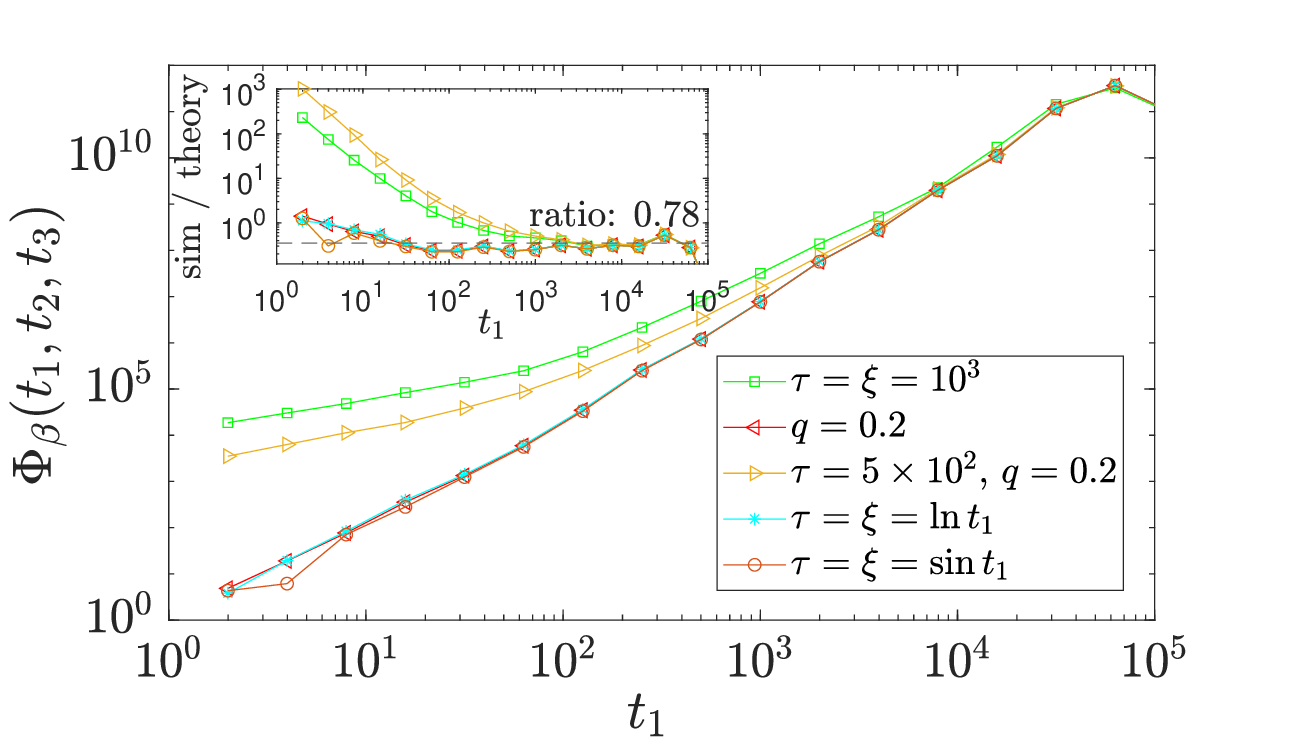}  
  \caption{different time lags for $\beta=0.9$}
  \label{fig:sub-second}
\end{subfigure}
%\begin{subfigure}{.35\textwidth}
%  \centering
%  % include third image
%  \includegraphics[width=.95\linewidth]{./Figs_Tayyab/h_fix_several_b.eps}  
%  \caption{Put your sub-caption here}
%  \label{fig:sub-third}
%\end{subfigure}
\caption{\label{fig:running-t_1}These figures show the log-log plots for $\beta = 0.3$ (left panel) and $\beta = 0.9$ (right panel) of the 3-point PACF $\Phi_\beta(t_1, t_1 + \tau, t_1 + \tau + \xi)$ as a function of $t_1$, where the time lags $\tau$ and $\xi$ are defined in Eq. \eqref{eq.19}. The data are obtained for a wide range of combinations of $\tau$ and $\xi$, as displayed in the legends. Power law behavior in the last few decades, determined by the slope of the best fit to the data. The two insets show the ratio of the numerically estimated LLg correlation $\Phi_\beta$ (\cf~Eq. \eqref{eq:allrun_pwrlaw_LLg}) to the correlation scaling $\phi_\alpha$ (\cf~Eq. \eqref{eq:allrun_pwrlaw_SM}) of the SM, which is computed analytically.
%In figure(a) the correlation function are taken on the fixed value of constant splitting times $q,\tau$ as mentioned in the graph, and the LLg parameter($\beta$) varies and mentioned in this figure as legend. The slopes of data fit line on these values, through the data points mentioned in the graph shows the power law behaviors for the LLg as exponents $(E_7(0.1),E_7(0.3),E_7(0.6),E_7(0.8),E_7(0.9))$. On the other hand, in figure(b) the correlation function are produced on the fixed value of $\beta$ and the constant splitting times values are varying as mentioned in the legends along with the slope of data lines as exponents say $(E_7(0.1a), E_7(0.1b), E_7(0.1c), E_7(0.1d))$. This is named as the asymptotic scaling of $3$-point PACF for the LLg.
}
%\label{fig:fig}
\end{figure}
%%%%%%%%%%%%%%%%%%%%%%%%%%%%%%%%%%%%%%%%%%%%%%%%%%%%%%%%%%%%%%%%%%%%%%%%%%%%%%%%%%%%%%%%%%%%%%%%
\subsubsection{Scaling test for  $\Phi_\beta(t_{1},t_{1}+\tau,t_{1}+\tau+\xi)$ with fixed time lag $\tau,\xi >0$}
In Lemma \ref{lem:univrsl-scaling}, the scalings of the $3$-point PACF of the SM were derived. For \(0 < \alpha < 3\), with \(\tau\) and \(\xi\) are defined in Eq. \eqref{eq.19}, the scaling behavior of the SM correlation can be expressed as  
\begin{equation}\label{eq:allrun_pwrlaw_SM}  
\phi_\alpha(m_{1}, m_{1}+\tau, m_{1}+\tau+\xi) \sim \frac{6}{3-\alpha} m_{1}^{3-\alpha}, \quad \text{as} \quad m_{1} \to \infty.  
\end{equation}  
Similarly, the LLg correlation data are collected for specific values of \(\beta\), with some time lags \(\tau\) and \(\xi\), as indicated in Eq.~\eqref{eq.19}. The LLg correlation  
$\Phi_\beta(t_{1}, t_{1}+\tau, t_{1}+\tau+\xi)$  
is plotted as a function of \(t_1\) in Fig. \ref{fig:running-t_1}. Results for $\beta=0.3$ and $0.9$ with several time lags $\tau$ and $\xi$ are presented in Figs. \ref{fig:sub-first} and \ref{fig:sub-second}, respectively. The slopes of the data fit lines for the last few decades reveal power-law behavior for the LLg, with exponents denoted by $\nu_3$.  
This indicates that the LLg correlation exhibits asymptotic scaling behavior for large \(t_1\). Hence, the PACF of the LLg follows
\begin{eqnarray}\label{eq:allrun_pwrlaw_LLg}  
\Phi_{\beta} (t_1,t_1+\tau,t_1+\tau+\xi) \sim K\, t_{1}^{\nu_3},\qquad t_1 \rightarrow \infty.  
\end{eqnarray}  
As in previous cases, this behavior also follows the \((\alpha,\beta)\) dependence when these parameters are properly tuned (\cf~Eq. \eqref{eq.39} for \(p=2\)). We observe that for both $\beta$ values, the data collapse onto a single line for all time-lag relationships defined in Eq. \eqref{eq.19}. This supports the SM correlation prediction, indicating that the dependence is solely on $\alpha$. The last few decades of data in Fig. \ref{fig:running-t_1} further confirm this. The insets of Fig. \ref{fig:running-t_1} shows the ratio of the numerical LLg correlation data to the theoretical correlation expression of the SM. The small deviation confined to the interval \((0,1)\), hallmark the accuracy of the numerical results.

Figure~\ref{fig:varing_beta_diff_lags} illustrates how the SM and the LLg dynamics yield the same scaling behavior for the $3$-point PACF in the superdiffusive regime.  In Fig.~\ref{fig:h_fix_several_b}, the main panel shows $\Phi_\beta(t_1,t_2,t_3)$ on a log–log scale for various $\beta$, exhibiting excellent power‐law scaling that matches the SM prediction across all $\beta$. The lower subpanel plots the ratio of simulation to theory, which stays close to one for all cases, indicating remarkable agreement between simulation and theory. Fig.~\ref{fig:q04_several_b} repeats this test by allowing $\tau$ and $\xi$ to grow as $t_1^{0.4}$; the upper panel shows that the numerical data for all $\beta$ follow the theoretical power‐law trend, and the simulation/theory ratio in the lower subpanel remains uniformly near unity, again confirming that the deterministic and stochastic models produce indistinguishable PACF scaling. Thus, Figure~\ref{fig:varing_beta_diff_lags} provides direct evidence that, in the limit of large and well-separated times, the $3$-point PACF in both SM and LLg follows the same superdiffusive scaling law, supporting the conclusion that a common scaling law governs deterministic and stochastic transport in this regime.

\begin{figure}[t] %ht
\hspace{-12mm}
\begin{subfigure}{.5\textwidth}
  \centering
  % include first image
  \includegraphics[width=1.1\linewidth]{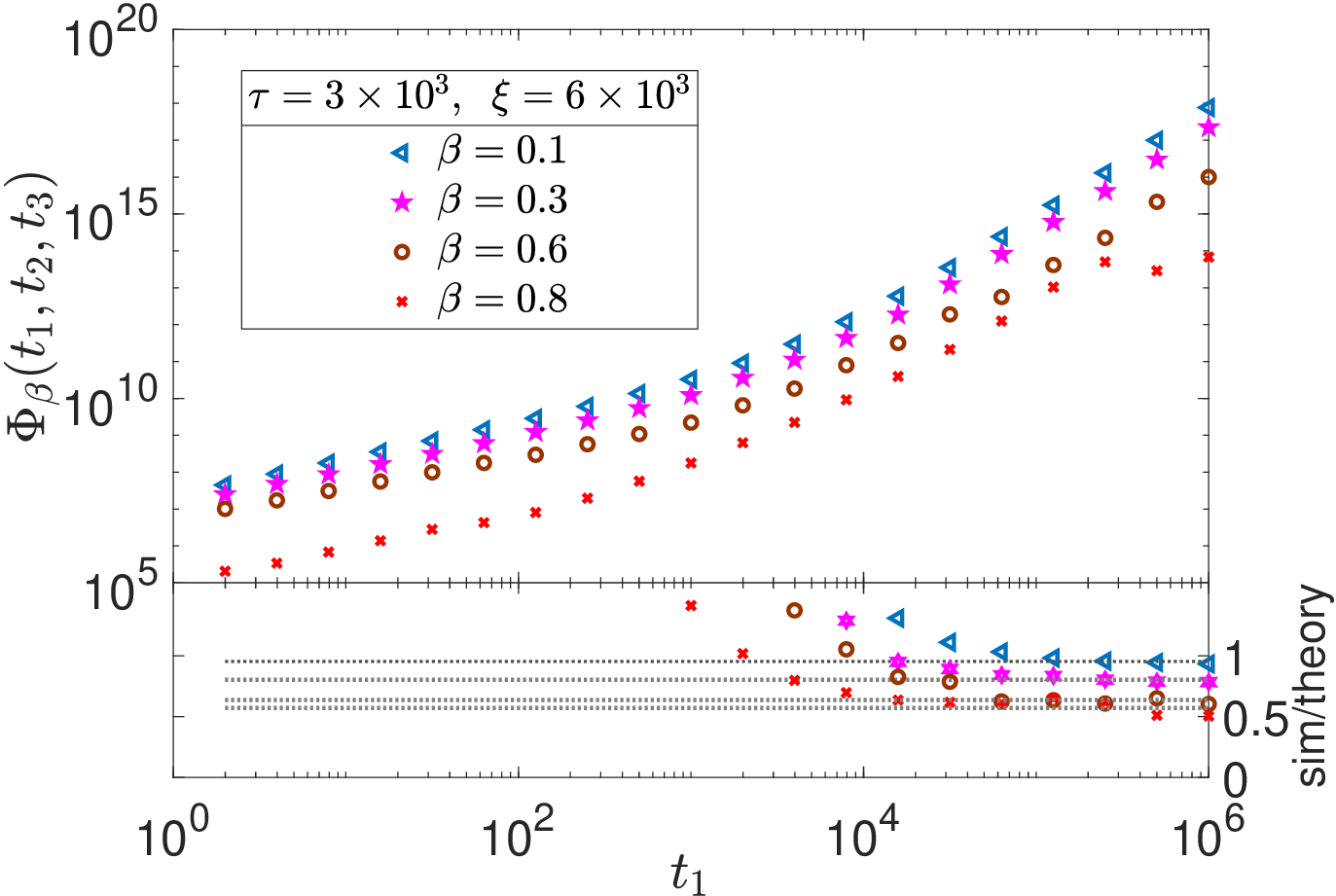}  
  \caption{varying $\beta$}
  \label{fig:h_fix_several_b}
\end{subfigure}
\hspace{0.55cm}
\begin{subfigure}{.5\textwidth}
  \centering
  % include second image
  \includegraphics[width=1.11\linewidth]{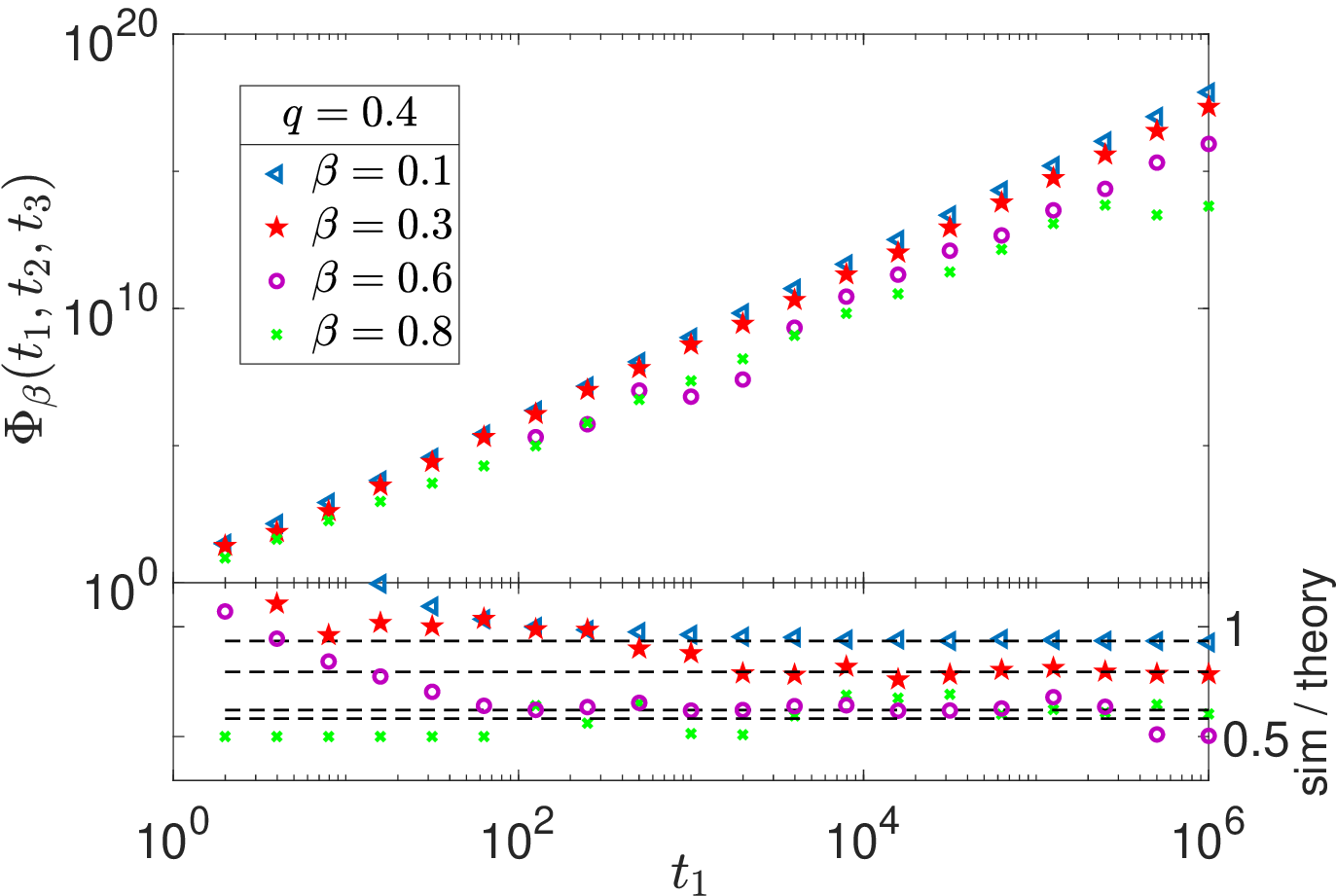}  
  \caption{varying $\beta$}
  \label{fig:q04_several_b}
\end{subfigure}
\caption{\label{fig:varing_beta_diff_lags}
%  Log-log plot of the $3$-point PACF $\Phi_\beta(t_1,t_2,t_3)$ of the LLg versus $t_1$ for four values of $\beta$ (see legends). In Fig. \ref{fig:h_fix_several_b}, the time lags are fixed at $\tau = 3\times10^3$ and $\xi = 6\times10^3$. Different markers denote the numerical simulation data. 
%The lower subpanel (semilog scale) shows how the simulation data compares to the theoretical curve, with a dotted horizontal line indicating the average value of the normalized data in the last decades of the horizontal axis, which stays close to unity for all $\beta$, indicating a strong match. Fig. \ref{fig:q04_several_b} does the same as Fig. \ref{fig:h_fix_several_b} but with the time lag scaled as $\tau=\xi \sim t_1^q$ (with exponent $q = 0.4$) for both lags. The lower subpanel (semilog scale) again shows the sim/theory ratio and the horizontal dotted average line, which remain close to unity for all $\beta$, confirming strong agreement. These panels demonstrate that both with fixed and with time-dependent lags, the simulated $3$-point PACF collapses onto the analytically predicted power laws for all $\beta$.
%
Log-log plot of the $3$-point PACF $\Phi_\beta(t_1,t_2,t_3)$ versus $t_1$ for  varying $\beta$ (see legends). In Fig.~ \ref{fig:h_fix_several_b}, time lags $\tau=3\times10^3$ and $\xi=6\times10^3$ are fixed; in Fig. \ref{fig:q04_several_b}, $\tau=\xi\sim t_1^{q}$, where $q=0.4$. Different markers denote numerical simulation data of the LLg correlation. The lower subpanels are on semilog scale and display the ratio \emph{simulation/theory} with a horizontal dotted line at the average value over the last decade of $t_1$, which remains close to unity for all $\beta$, indicating a strong match. These panels demonstrate that both with fixed and with time-dependent lags, the simulated $3$-point PACF collapses onto the analytically predicted power laws for all $\beta$.}
\end{figure}
%%%%%%%%%%%%%%%%%%%%%%%%%%%%%%%%%%%%%%%%%%%%%%%%%%%%%%%%%%%%%%%%%%%%%%%%%%%%%%%%%%%%%%%%%%%%%%%%
\section{Discussion and conclusion}\label{sec.5}
Position moments and correlations often explain the transport properties of anomalous diffusion, leading to different diffusive regimes. 
The transition from normal to anomalous diffusion is described in the framework of generalized diffusion dynamics. 
Strong anomalous superdiffusion is characterized by the position moments and correlations dominated by ballistic trajectories, and it depends on the transport exponent $\gamma$. This results in two-part linear scaling of $\gamma(p)$, as shown in Eq. \eqref{eq:scale-invrt-momnts}.

In this paper, we investigated the scaling behavior of higher-order PACFs in deterministic and stochastic systems, focusing on the SM and the LLg. First, we analytically derived the generalized PACF $\phi_\alpha$ of the SM and some of its scaling forms, following the relationships between $j$ times (see Lemma \ref{lem:jpoint-scaling}). For the special case $j=3$, we analytically derived the $3$-point PACF $\phi_\alpha(m_{1},m_{2},m_{3})$ and some of its scaling forms to determine the power-law behavior for different relationships between $m_{1}$, $m_{2}$, and $m_{3}$.
Based on the relationship in Eq. \eqref{eq.39}, we analyzed the interplay between the power-law exponents of these two dynamics. 
%Our analysis shows a strong overlap in the correlation scaling properties of these systems, particularly in the strongly superdiffusive regime. 
%The remarkable equivalence of these two dynamics holds for small values of $\beta$. 
Our analysis shows a strong overlap in the correlation scaling properties of these systems, particularly in the strongly superdiffusive regime. The practical significance of studying anomalous diffusion via deterministic systems like the SM lies in bridging theoretical insights with real-world applications. Deterministic dynamics offer analytical tractability for complex transport phenomena observed in intracellular transport in crowded biological environments \cite{MP24}, dispersion in porous media \cite{F75}, and geophysical flows \cite{B02}, where stochastic simulations are computationally prohibitive. By replicating stochastic behaviors e.g., LLg correlations, the SM enables efficient prediction of scaling laws and universal transport properties, aiding in biomedical engineering, material design, and environmental modeling. This equivalence simplifies the interpretation of experimental data and advances predictive tools for systems where anomalous diffusion governs functionality. The remarkable equivalence of these two dynamics holds for small values of $\beta$. However, as the LLg parameter $\beta$ increases, the equivalence with the SM becomes less convincing because when the $\beta$ increases, the probability distribution of the scatterers yields a greater exponent, resulting in denser scatterers and more pronounced deviations in the system behavior, reducing consistency with the SM. Therefore,
by analytically deriving the scaling of the $3$-point PACF for the SM and numerically estimating their counterparts for the LLg, we established their equivalence across various transport regimes.
In conclusion, our findings show a strong overlap between the SM and LLg dynamics, not only at the level of all position moments but also in the scaling of correlation functions, providing deep insights into anomalous diffusion processes. These findings adds to our theoretical knowledge and call into question existing models, emphasizing the importance of additional research on alternative deterministic and stochastic dynamics. Despite the limitation of our study, which include specific modeling assumptions, the implications for practical applications in domains such as statistical mechanics and biological systems are significant. Future studies might focus on applying these findings to other complex systems to enhance the discussion of anomalous transport. 
%Finally, this work emphasizes the need to investigate the interaction between deterministic and stochastic dynamics, which will lead to deeper insights into the behavior of complex systems.
%This work provides a template for translating stochastic correlation features into analytically tractable deterministic dynamics, simplifying experimental data interpretation in crowded cellular environments and porous materials.
Finally, this work emphasizes the need to investigate the interaction between deterministic and stochastic dynamics, and provides a template for translating stochastic correlation features into analytically tractable deterministic dynamics, simplifying experimental data interpretation in crowded cellular environments and porous materials, which will and leading to deeper insights into the behavior of complex systems.
\section*{Acknowledgement}
M.T. gratefully acknowledges the computational
resources provided by HPC@POLITO and, the project
for Academic Computing of the Department of Control and Computer Engineering at the Politecnico di Torino, Italy \cite{HPC_PoliTo}.

\section*{References}
%abbrv
\bibliographystyle{iopart-num}
\bibliography{bibliography}
%\begin{thebibliography}{9}
%\bibitem{iopartnum} IOP Publishing is to grateful Mark A Caprio, Center for Theoretical Physics, Yale University, for permission to include the {\tt iopart-num} \BibTeX package (version 2.0, December 21, 2006) with  this documentation. Updates and new releases of {\tt iopart-num} can be found on \verb"www.ctan.org" (CTAN). 
%\end{thebibliography}

\end{document}